\documentclass[11pt, letterpaper]{article}
\usepackage{standalone}
\usepackage{etoolbox}

\usepackage{nicefrac}
\usepackage[letterpaper,margin=1in]{geometry}
\usepackage{amsthm}
\usepackage{amsmath}
\usepackage{amssymb}
\usepackage{thmtools}
\usepackage{thm-restate}
\usepackage{mathtools}
\usepackage{mdframed}

\usepackage{setspace}

\usepackage{placeins}

\usepackage[table]{xcolor}

\definecolor{linkcol}{rgb}{0,0,0.38}
\definecolor{citecol}{rgb}{0.8,0,0}
\definecolor{urlcol}{rgb}{0.1,0.35,0}

\usepackage[pdfpagelabels,bookmarks=false,hyperfootnotes=false,hypertexnames=false]{hyperref}
\hypersetup{colorlinks, linkcolor=linkcol, citecolor=citecol, urlcolor=urlcol}

\allowdisplaybreaks

\DeclareFontFamily{U}{BOONDOX-calo}{\skewchar\font=45 }
\DeclareFontShape{U}{BOONDOX-calo}{m}{n}{
  <-> s*[1.05] BOONDOX-r-calo}{}
\DeclareFontShape{U}{BOONDOX-calo}{b}{n}{
  <-> s*[1.05] BOONDOX-b-calo}{}
\DeclareMathAlphabet{\link}{U}{BOONDOX-calo}{m}{n}
\DeclareMathAlphabet{\blink}{U}{BOONDOX-calo}{b}{n}

\usepackage[utf8]{inputenc}

\usepackage[
backend=biber,
style=alphabetic,
citestyle=alphabetic,
maxalphanames=4,
maxcitenames=99,
mincitenames=98,
maxbibnames=99,
giveninits=true,
]{biblatex}

\usepackage[nameinlink]{cleveref}

\usepackage{lmodern}

\usepackage{bbm}
\usepackage{thm-restate}

\newtheoremstyle{light} 
    {\topsep}                    
    {\topsep}                    
    {\itshape}                   
    {}                           
    {\scshape}                   
    {.}                          
    {.5em}                       
    {}  

\newtheorem{theorem}{Theorem}[section]
\newtheorem{lemma}[theorem]{Lemma}

\newtheorem{definition}[theorem]{Definition}

\newtheorem{observation}[theorem]{Observation}

\theoremstyle{light}

\makeatletter
\if@cref@capitalise
\crefname{claiminproof}{Claim}{Claims}
\else
\crefname{claiminproof}{claim}{claims}
\fi

\if@cref@capitalise
\crefname{algocf}{Algorithm}{Algorithms}
\else
\crefname{algocf}{algorithm}{algorithms}
\fi

\if@cref@capitalise
\crefname{conjecture}{Conjecture}{Conjectures}
\else
\crefname{conjecture}{conjecture}{conjectures}
\fi

\if@cref@capitalise
\crefname{thm}{Theorem}{Theorems}
\else
\crefname{thm}{theorem}{theorems}
\fi

\if@cref@capitalise
\crefname{lem}{Lemma}{Lemmas}
\else
\crefname{lem}{lemma}{lemmas}
\fi

\makeatother

\usepackage{url}
\usepackage{array}

\usepackage{setspace}

\usepackage{wrapfig}

\makeatletter
\newcommand{\labeltarget}[1]{\Hy@raisedlink{\hypertarget{#1}{}}}
\makeatother

\usepackage{upgreek}

\usepackage{complexity}

\usepackage[inline]{enumitem}
\usepackage{moreenum}

\usepackage[super]{nth}

\usepackage{bm}

\usepackage{xspace}

\usepackage{ifthen}

\usepackage[immediate]{silence}
\usepackage{sectsty}
\DeactivateWarningFilters[tmp]
\allsectionsfont{\boldmath}

\setlist[enumerate]{nosep,topsep=0.1em}
\setlist[enumerate,1]{label=(\roman*), leftmargin=2.2em}
\setlist[itemize]{nosep,topsep=0.3em}

\usepackage[textsize=footnotesize, color=blue!30!white]{todonotes}
\usepackage{xfrac}

\usepackage{tikz}
\usetikzlibrary{calc}
\usetikzlibrary{math}
\usetikzlibrary{shapes.geometric}
\usetikzlibrary{arrows}
\usetikzlibrary{decorations.pathreplacing, decorations.pathmorphing}

\usepackage{tcolorbox}
\tcbuselibrary{skins}

\usepackage{float}
\usepackage{graphicx}
\graphicspath{{../graphics/}}
\makeatletter
\newcommand\appendtographicspath[1]{%
  \g@addto@macro\Ginput@path{#1}%
}
\makeatother
\usepackage[margin=10pt,font=small,labelfont=bf,skip=-5pt]{caption}
\usepackage{subcaption}

\usepackage[linesnumbered,vlined,algo2e]{algorithm2e}
\SetAlgoSkip{bigskip}
\SetAlgoInsideSkip{smallskip}

\usepackage{mdframed}

\let\truehypersetup\hypersetup
\renewcommand\hypersetup[1]{}
\usepackage{bigfoot}
\let\hypersetup\truehypersetup
\DeclareMathOperator{\argmax}{argmax}

\DeclareMathOperator{\spn}{span}
\DeclareMathOperator{\rank}{rank}
\csundef{poly}

\renewcommand{\epsilon}{\varepsilon}

\newcommand{\GreedySelectionRule}{\textsc{GreedySelectionRule}\xspace}
\newcommand{\MinimalLink}{\textsc{MinimalLinkConstruction}\xspace}
\newcommand{\SingleOcrsLink}{\textsc{SingleOcrsLink}\xspace}

\definecolor{green}{rgb}{0.4,0.85,0.6}

\newbool{sodasubm}
\setbool{sodasubm}{false}

\title{Nearly Tight Sample Complexity for Matroid\\Online Contention Resolution}

\ifbool{sodasubm}{
  \author{}
  \date{}
}{
\author{
Moran Feldman\thanks{
Department of Computer Science, University of Haifa, Haifa, Israel.
Email: \href{mailto:moranfe@cs.haifa.ac.il}%
{moranfe@cs.haifa.ac.il}.
}
\and
Ola Svensson\thanks{
School of Computer and Communication Sciences, EPFL, Lausanne, Switzerland.
Email: \href{mailto:ola.svensson@epfl.ch}%
{ola.svensson@epfl.ch}.
}
\and
Rico Zenklusen\thanks{
Department of Mathematics, ETH Zurich, Zurich, Switzerland.
Email: \href{mailto:ricoz@ethz.ch}%
{ricoz@ethz.ch}.}
}
\date{}
}

\begin{document}

\maketitle
\thispagestyle{empty}
\addtocounter{page}{-1}

\begin{abstract}
Due to their numerous applications, in particular in Mechanism Design, Prophet Inequalities have experienced a surge of interest.
They describe competitive ratios for basic stopping time problems where random variables get revealed sequentially.
A key drawback in the classical setting is the assumption of full distributional knowledge of the involved random variables, which is often unrealistic.
A natural way to address this is via sample-based approaches, where only a limited number of samples from the distribution of each random variable is available.
Recently, Fu, Lu, Gavin Tang, Wu, Wu, and Zhang (2024) showed that sample-based Online Contention Resolution Schemes (OCRS) are a powerful tool to obtain sample-based Prophet Inequalities.
They presented the first sample-based OCRS for matroid constraints, which is a heavily studied constraint family in this context, as it captures many interesting settings.
This allowed them to get the first sample-based Matroid Prophet Inequality, using $O(\log^4 n)$ many samples (per random variable), where $n$ is the number of random variables, while obtaining a constant competitiveness of $\sfrac{1}{4}-\varepsilon$.

We present a nearly optimal sample-based OCRS for matroid constraints, which uses only $O(\log \rho \cdot \log^2\log\rho)$ many samples, almost matching a known lower bound of $\Omega(\log \rho)$, where $\rho \leq n$ is the rank of the matroid.
Through the above-mentioned connection to Prophet Inequalities, this yields a sample-based Matroid Prophet Inequality using only $O(\log n + \log\rho \cdot \log^2\log\rho)$ many samples, and matching the competitiveness of $\sfrac{1}{4}-\varepsilon$, which is the best known competitiveness for the considered almighty adversary setting even when the distributions are fully known.
\end{abstract}

\newpage

\section{Introduction}

The classical prophet inequality, introduced and solved by Krengel, Sucheston, and Garling~\cite{krengelSemiamartsFiniteValues1977,krengelSemiamartsAmartsProcesses1978}, is a cornerstone of optimal stopping theory.%
\footnote{Garling, although not a co-author of these seminal papers, is credited by Krengel and Sucheston for providing a result that leads to the tight $\nicefrac{1}{2}$-competitiveness of the prophet inequalities.
This is why Garling is typically credited in this context.}
In this problem, an agent observes a sequence of values drawn from independent known distributions, and must irrevocably decide at each step whether to accept the current value and stop, or discard this value and continue. The goal is to maximize the expectation of the single accepted value. The seminal result of~\cite{krengelSemiamartsAmartsProcesses1978} shows that a simple threshold-based strategy can achieve an expected value of at least $\nicefrac{1}{2}$ of what a ``prophet'', who knows all values in advance, could obtain. They also showed that this ratio is optimal.


In recent years, the deep connections between the prophet inequality problem and Algorithmic Game Theory have motivated generalizations beyond the single accepted value setting.
For example, applications in online auctions and pricing, such as the simple and practical posted-price auctions, often necessitate the selection of a \emph{set} of values (also known as items in this context).
This has led to the study of prophet inequalities under complex combinatorial constraints (see~\cite{hajiaghayiAutomatedOnlineMechanism2007,chawlaMultiparameterMechanismDesign2010a,correaPricingProphetsBack2019,lucierEconomicViewProphet2017} and references therein).
A powerful and widely studied example of such a prophet inequality generalization is the Matroid Prophet Inequality, introduced by \textcite{kleinbergMatroidProphetInequalities2012a}, where the set of chosen items must satisfy a matroid constraint. %
Matroids are a rich combinatorial constraint family that capture many interesting setting.\footnote{Formally, a matroid $M=(N,\mathcal{I})$ consists of a finite set $N$, called \emph{ground set}, and a nonempty family $\mathcal{I}\subseteq 2^N$ of subsets of $N$, called \emph{independent sets}, that satisfy
\begin{enumerate*}
  \item if $I\in \mathcal{I}$ and $J\subseteq I$, then $J\in \mathcal{I}$, and
  \item if $I,J\in \mathcal{I}$ with $|J|>|I|$, then there exists $e\in J\setminus I$ such that $I\cup \{e\}\in \mathcal{I}$
\end{enumerate*}.
We refer the interested reader to~\cite{oxleyMatroidTheory1992,welshMatroidTheory2010} and \cite[Volume B]{schrijverCombinatorialOptimizationPolyhedra2003} for further information on matroids.
}
Impressively, Kleinberg and Weinberg showed that the optimal $\nicefrac{1}{2}$-competitive ratio extends to all matroids.
However, a major caveat of all generalizations along this line is the strong assumption that the underlying value distributions are perfectly known.
To address this limitation, a newer sample-based model has been suggested (see \cite{azarProphetInequalitiesLimited2014}).
In this model, an algorithm does not know upfront the value distributions, and can access them only by sampling from them.
When the number of samples is large, this model recovers the classical setting in which the value distributions are known.
Thus, the focus in the sample-based model is on designing algorithms that use a small number of samples.

Prophet inequality problems are related to Online Contention Resolution Schemes (OCRS), which are a versatile tool introduced by \textcite{FeldmanSZ21} to model stochastic selection processes.
The OCRS setting for a matroid $M=(N,\mathcal{I})$ is defined as follows: given an input vector $x\in [0,1]^N$ in the matroid polytope (i.e., $x$ represents the marginals of some distribution over sets in $\mathcal{I}$), nature draws a random \emph{active} set $R(x) \subseteq N$ that includes each element $e \in N$ independently with probability $x_e$ (notice that $R(x)$ is drawn from a product distribution; in the following, we refer to this distribution as $\mathcal{D}(x)$).
After $R(x)$ is drawn by nature, the elements of $R(x)$ arrive one-by-one in an adversarial order, and, upon each arrival, an algorithm must make an irrevocable decision to either select the element or discard it, while keeping the set of accepted elements independent in the matroid $M$ at all times. The performance of an OCRS is measured by its \emph{$c$-selectability}, which is the minimum probability that an element $e$ is selected, conditioned on it being active (i.e., $e \in R(x)$).

OCRSs have found numerous applications, and have also been considered for various other combinatorial constraints beyond matroids, for other objectives, non-independent distributions, and also for weaker adversaries (see~\cite{%
FeldmanSZ21,%
adamczykRandomOrderContention2018,%
leeOptimalOnlineContention2018a,%
pollnerImprovedOnlineContention2024,%
guptaPairwiseIndependentContentionResolution2024,%
macruryRandomOrderOnlineContention2025,%
yoshinagaOnlineContentionResolution2025%
} and references therein).
A typical way to link them to prophet inequalities is by designing threshold-based algorithms that select for each element $e \in N$ a threshold weight $\tau_e \in \mathbb{R}_{\geq 0}$ above which an element is considered for selection.
For a given set of thresholds (over which one needs to optimize), the problem can be translated to the one of designing an OCRS, where an element is considered to be active if its revealed value is above the threshold.

\textcite{FuFLLWWZ24} showed that also sample-based prophet inequalities can be constructed through the framework of OCRSs, namely by designing sample-based OCRSs.
In a sample-based OCRS, the vector $x$ is unknown, and the algorithm is only given access to $k$ independent samples from $\mathcal{D}(x)$ (where $k$ should be chosen to be the smallest value that still allows for obtaining the desired selectability).
In this problem, we assume the strongest type of adversary, known as the \emph{almighty adversary}, who can choose the order of the elements as a function of both the realization of $R(x)$ and the randomness of the algorithm.
The reduction provided in \cite{FuFLLWWZ24} shows that a $c$-selectable sample-based OCRS can be transformed into a sample-based Matroid Prophet Inequality that is $(c-\varepsilon)$-competitive.
This works for any $\varepsilon >0$, as long as one has at least $O_{\varepsilon}(\log n)$ samples (per element), where $n\coloneqq |N|$ is the size of the ground set $N$.
The $O_{\varepsilon}(\log n)$ samples are needed for the reduction, and if the sample-based OCRS itself uses $k$ samples, then the resulting sample-based Matroid Prophet Inequality algorithm uses $O_{\varepsilon}(\log n) + k$ many samples altogether.
Leveraging this reduction, \cite{FuFLLWWZ24} presented the first constant-competitive sample-based Matroid Prophet Inequality, by developing a $(\nicefrac{1}{4}-\varepsilon)$-selectable OCRS using $O(\log^4 n)$ samples.

In this work, we significantly improve the sample complexity for sample-based matroid OCRS, obtaining the following bound through a randomized algorithm. (We use $\rank(M)$ to denote the rank of the matroid $M$, i.e., the size of the largest independent set in it.)
\begin{theorem}\label{thm:main}
    For any $\varepsilon>0$, there is an polynomial-time sample-based randomized online contention resolution scheme for any matroid $M$ with a selectability of $\nicefrac{1}{4}-\varepsilon$ against the almighty adversary. This algorithm uses $O(\log \rank(M) \cdot \log^2 \log\rank(M))$ many samples.
\end{theorem}
 Through the reduction in~\cite{FuFLLWWZ24}, our work immediately yields a $(\nicefrac{1}{4}-\varepsilon)$-competitive Matroid Prophet Inequality algorithm with the above sample complexity plus $O(\log n)$ samples that are incurred by the reduction itself. 

Our result is tight in two dimensions up to small factors. First, the sample complexity is nearly optimal. As was shown in~\cite{FuFLLWWZ24}, even an \emph{offline} contention resolution scheme requires $\Omega(\log n)$ samples  (for a matroid $M$ with $\log n \approx \log \rank(M)$) to achieve any constant selectability, leaving only a small gap of $\log^2 \log \rank(M)$ compared to our upper bound. Second, the selectability of $\nicefrac{1}{4}$ is the best known factor against an almighty adversary, even in the classic setting where the probability vector $x$ is fully known~\cite{FeldmanSZ21}.
One should note, however, that while the number of samples used by our algorithm is nearly tight for sample-based matroid OCRSs, no similar lower bounds are known for sample-based Matroid Prophet Inequality.
Indeed, the sample complexity for the last problem is highly related to the notorious Matroid Secretary Problem: a lower bound showing that more than a single sample is necessary to get a constant competitive ratio for sample-based Matroid Prophet Inequality would immediately rule out the possibility of getting such a competitive ratio for the Matroid Secretary Problem using a class of algorithms known as ``order-oblivious algorithms'' (introduced by~\cite{azarProphetInequalitiesLimited2014}).
Such a result would be a major step towards refuting the famous Matroid Secretary Conjecture~\cite{BabaioffIKK18}, which has been open for over $15$ years, and claims that the Matroid Secretary Problem enjoys an algorithm with a constant competitive ratio.
For more information on the Matroid Secretary Problem, and its connections to Contention Resolutions Schemes, we refer the interested reader to \cite{%
BabaioffIKK18,%
dughmiMatroidSecretaryEquivalent2022,%
lachishOlogLogMathrmRank2014,%
feldmanSimpleOlogLogrankCompetitive2015%
} and references therein.

\paragraph{Outline.} The remainder of this paper is organized as follows.
In \Cref{sec:definitionOverview}, we introduce key OCRS concepts of \cite{FeldmanSZ21} for the known $x$ setting, and how OCRSs can be constructed by decomposing the matroid through a chain.
\Cref{sec:formalAlgorithm} details our main contribution: the novel sample-based OCRS algorithm, which also uses a (sample-based) chain decomposition.
We first present our core subroutine, which shows how to construct a single link of the chain using samples and a carefully randomized number of steps.
Then, we show how this procedure can be iterated to build a full spanning chain, which then defines our algorithm.
The core technical analysis of our algorithm is presented in \Cref{sec:progress,sec:inlink_loss}, where we formally bound the error probabilities and prove that our randomized procedure makes sufficient progress in each step, which ensures that the construction succeeds with high probability.

\section{OCRS for Known $x$ via Spanning Chains}
\label{sec:definitionOverview}

In this section, we introduce key concepts from the OCRS of~\cite{FeldmanSZ21} for the setting of a known vector $x$. 
This provides context for our definitions, and sets the stage for the modifications detailed in \Cref{sec:formalAlgorithm}, where we formally describe our algorithms.
Let us begin with a central combinatorial object termed \emph{spanning chain}, which is a laminar family of sets that includes both $N$ and $\emptyset$ as members of the family.
\begin{definition}[spanning chain]
  A \emph{spanning chain} of a finite set $N$ is a tuple $\mathcal{C}=(C_0, \dots, C_k)$ with $N=C_0 \supseteq C_1 \supseteq \dots \supseteq C_k = \emptyset$.
The sets $C_i$ in a spanning chain are called \emph{links}.
\end{definition}

The number $k$ of links in the spanning chain can vary, and we refer to it as the \emph{length} of the spanning chain. A spanning chain is used to guide the selection decisions of the OCRS of~\cite{FeldmanSZ21} (but \cite{FeldmanSZ21} refer to this chain as ``chain decomposition''). The spanning chain used by~\cite{FeldmanSZ21} arises naturally from an iterative procedure that gets as input a matroid $M=(N,\mathcal{I})$, a vector $x\in P_M$, where $P_M\subseteq [0,1]^N$ is the matroid polytope of $M$,
and a parameter $\tau \in (0,1)$. Intutively, the parameter $\tau$ upper bounds the maximum probability that an element is ``spanned'' (i.e., cannot be picked due to the selection of other elements).

The iterative procedure starts by setting $C_0 = N$. Subsequent links of the spanning chain are constructed by identifying elements that are likely to be spanned. For example, the logic used to construct the link $C_1$ from $C_0$ is specified by the following frame, where $\spn(S)$ denotes the span of a set $S$ in the matroid $M$, i.e., $\spn(S) = \{e\in N \colon r(S \cup \{e\}) = r(S)\}$, where we denote by $r(S)=\max\{|I|\colon I\in \mathcal{I}, I\subseteq S\}$ the rank function of $M$.
\begin{center}
\begin{minipage}{0.95\textwidth}
\begin{mdframed}[hidealllines=true, backgroundcolor=gray!20]
\MinimalLink to find $C_1$ from $C_0=N$:
    \begin{itemize}
         \item Initialize $A_0 = \emptyset$.
    \item Assuming we have defined $A_0, A_1, \dots, A_{i-1}$, let
    \[ A_i = \{e \in C_0 : \Pr[e \in \text{span}((R(x) \cup A_{i-1}) \setminus \{e\})] > \tau \} . \]
    That is, $A_i$ contains those elements that are likely to be spanned assuming that the elements of $A_{i-1}$ are contracted.\footnotemark
\end{itemize}
Notice that $A_{i-1} \subseteq A_i$ for every $i \ge 1$, and $A_i = A_{i-1}$ implies $A_j = A_i$ for every $j > i$. Thus, the sequence of sets is guaranteed to stabilize after at most $|C_0|$ iterations. We define the second link $C_1$ of the spanning chain to be this final (stabilized) set.
\end{mdframed}
\end{minipage}
\end{center}
\footnotetext{Intuitively, contracting elements in a matroid means that one views these elements as implicitly being part of every set. More formally, contracting a set $A$ in the matroid $M$ produces a new matroid over the ground set $N\setminus A$ whose independent sets are $\{S \subseteq N \setminus A : r(S \cup A) = r(S) + r(A)\}$.}

The construction of the next links of the spanning chain is done similarly. Specifically, $C_2$ is generated from $C_1$ by applying the same logic to the matroid obtained by restricting $M$ to the set $C_1$,\footnote{Restricting a matroid $M=(N,\mathcal{I})$ to a set $C\subseteq N$ means removing all elements that do not belong to $C$ from the ground set of $M$. The resulting matroid is denoted by $M|_C$.} $C_3$ is generated from $C_2$ by applying the same logic to the matroid obtained by restricting $M$ to $C_2$, and so on. This construction process continues until no element is likely to be spanned, i.e., an empty link $C_k = \emptyset$ is generated. The main technical observation of~\cite{FeldmanSZ21} in this context was that this iterative construction process is guaranteed to eventually generate such an empty link when $x\in \lambda \cdot P_M$ for some $\lambda \in [0,1)$ and $\tau > \lambda$.

The spanning chain $\mathcal{C} = (C_0, C_1, \ldots, C_k)$ resulting from the above construction defines a natural OCRS that is based on the following greedy selection rule.
\begin{center}
\begin{minipage}{0.95\textwidth} 
\begin{mdframed}[hidealllines=true, backgroundcolor=gray!20]
   \hypertarget{alg:greedySelectionRule}{\GreedySelectionRule:}
\begin{itemize}
    \item Upon arrival of an active element $e \in N$, let $i$ be the unique index with $e\in C_{i}\setminus C_{i+1}$. 
		\item Let $M_i$ be the matroid obtained from $M$ by restricted to $C_{i}$ and contracting $C_{i+1}$.
    \item Accept $e$ if and only if adding $e$ to the set of elements of $C_i\setminus C_{i+1}$ already accepted preserves independence in $M_i$. 
\end{itemize}
\end{mdframed}
\end{minipage}
\end{center}
\newcommand{\GreedySelectionRuleLink}[1][]{\hyperlink{alg:greedySelectionRule}{\GreedySelectionRule}\xspace}

It is important to note that the greedy selection rule always produces an independent set in the matroid because the spanning chain decomposes the ground sets into parts that can be considered separately (see, e.g., Theorem 5.1 in~\cite{doi:10.1137/110852061}).

To analyze the selectability of the OCRS obtained in this way, it is convenient to have the following definition.

\begin{definition}[$\mathcal{C}$-freeness of an element]
  Let $M=(N,\mathcal{I})$ be a matroid, $x\in P_M$, and  $\mathcal{C} = (C_i)_{i=0}^k$ be a spanning chain of $N$.
  Consider an element $e\in N$, and let $i$ be the unique index for which $e\in  C_i \setminus C_{i+1}$. 
  Then, the \emph{$\mathcal{C}$-freeness of $e$} (with respect to $M$ and $x\in P_M$) is
  \begin{equation*}
    \Pr[e\not\in \spn(((R(x) \setminus \{e\})\cap C_{i})\cup C_{i+1})].
  \end{equation*}
\end{definition}
Note that the $\mathcal{C}$-freeness of $e$ can be interpreted as the probability that \GreedySelectionRuleLink accepts $e$ when the active elements are the elements of $R(x) \cup \{e\}$, and $e$ is the last of them to arrive. Since the probability of $e$ to be accepted by \GreedySelectionRuleLink can only increase when $e$ arrives earlier, the $\mathcal{C}$-freeness of $e$ is also a lower bound on the probability that $e$ is accepted by \GreedySelectionRuleLink when the active elements are the elements of $R(x) \cup \{e\}$, regardless of the order in which these elements arrive.

For the spanning chain $\mathcal{C}$ constructed by the above iterative process, the construction guarantees that the $\mathcal{C}$-freeness of each element $e\in N$ is at least $\tau$. The following definition gives an easy way to state such observations.
\begin{definition}[$c$-balanced spanning chain for $x$]
  Let $M=(N,\mathcal{I})$ be a matroid and let $x\in P_M$.
  A spanning chain $\mathcal{C}$ of $N$ is called \emph{$c$-balanced for $x$} if, for each element $e\in N$, its $\mathcal{C}$-freeness with respect to $M$ and $x$ is at least $c$.
\end{definition}

For the purpose of the OCRS we construct in this paper, which only has sample access to $x$, we will need to consider distributions over spanning chains that are $c$-balanced for some $x\in P_M$.
\begin{definition}[$c$-balanced spanning chain distribution for $x$]
  Let $M=(N,\mathcal{I})$ be a matroid and let $x\in P_M$.
  A distribution $\mathcal{S}$ over spanning chains is called \emph{$c$-balanced} for $x$ if, for each element $e\in N$,
	\[
		\E_{\mathcal{C} \sim \mathcal{S}}[\text{$\mathcal{C}$-freeness of $e$ with respect to $M$ and $x\in P_M$}] \geq c,
	\]
	i.e., the expected $\mathcal{C}$-freeness of $e$ with respect to $M$ and $x$, over a random spanning chain $\mathcal{C}$ drawn from $\mathcal{S}$, is at least $c$.
\end{definition} 

We end this section by formally connecting the $c$-balancedness of a spanning chain distribution to the selectiveness of the OCRS that can be derived from it via \GreedySelectionRuleLink.
\begin{theorem}
\label{thm:chain-to-ocrs}
Recall that the OCRS setting consists of a matroid $M=(N,\mathcal{I})$ and a vector $x$ in the matroid polytope $P_M$. Given an oracle that samples a single spanning chain from a distribution $\cS$ that is $c$-balanced for $\lambda x$ (for some $c,\lambda \in [0,1]$), one can implement a $\lambda \cdot c$-selectable OCRS for $M$ without any additional access to $x$.
\end{theorem}

\begin{proof}
The OCRS works as follows. Let $\mathcal{C} = (C_0, C_1, \ldots, C_k)$ be the spanning chain sampled from the distribution. Every active element $e$ is immediately and independently discarded with probability $1 - \lambda$. Otherwise, $e$ is fed to the greedy selection rule \GreedySelectionRuleLink, and is accepted if the rule accepts it. Notice that from the point of view of the rule, the set of active elements that it receives is distributed like $R(\lambda x)$. Thus, by the $c$-balancedness of the distribution from which $\mathcal{C}$ is drawn, each element $e$ is accepted with probability at least $c$ conditioned on it being fed to \GreedySelectionRuleLink. The theorem now follows by recalling that each active element is fed to \GreedySelectionRuleLink with probability exactly $\lambda$.
\end{proof}

\newcommand{\algbuildocrs}[1][]{\hyperref[algo:buildOCRS]{\textup{\textsc{OcrsChain}\ensuremath{#1}}}\xspace}
\section{OCRS with Sample Access to $x$}
\label{sec:formalAlgorithm}

Theorem~\ref{thm:chain-to-ocrs} shows that to construct a sample-based OCRS for a matroid $M$, it suffices to describe an oracle that has only sample access to the distribution $\mathcal{D}(x)$, but still manages to sample from some distribution over spanning chains that is $c$-balanced for the vector $\lambda x$. Since sample access to $\mathcal{D}(x)$ implies sample access also to $\mathcal{D}(\lambda x)$, one can rephrase the task of the oracle as follows. Given only sample access to $\mathcal{D}(x)$ for some vector $x \in \lambda \cdot P_M$, the oracle should sample from some distribution over spanning chains that is $c$-balanced for the vector $x$. The rest of this paper is devoted to the implementation of such an oracle.


As explained in \Cref{sec:definitionOverview}, in the known $x$ setting, the oracle can return a chain whose links are constructed using the \MinimalLink procedure. This procedure relies on repeatedly testing conditions of the form $\Pr[e\in \spn(\cdot)]>\tau$. 
When $x$ is unknown, these spanning probabilities cannot be computed exactly. The natural alternative is to estimate these probabilities using samples drawn from $\mathcal{D}(x)$.
This sampling-based approach, however, introduces the risk of estimation errors. 
To bound the overall error probability, we must control the total number of estimations performed. 
Fortunately, the method of~\cite{FeldmanSZ21} provides a crucial starting point for such a bound: the spanning chain obtained by this method has a length of only $O(\log \rank(M))$ under the assumption $\lambda + \epsilon \leq \tau$ for any constant $\varepsilon > 0$.

This logarithmic chain length suggests a path toward a construction using a small number of samples. 
Specifically, if constructing each link required only a few probability estimations, then we would achieve a low total sample complexity for the construction of the entire chain.
For example, if each link demanded only a constant number of estimations (per element), a total of $O(\log \rank(M))$ estimations would be needed for constructing the entire chain. 
By standard Chernoff bounds, the error probability for each estimation can then be made as small as $1/\text{polylog}(\rank(M))$ using $O(\log \log \rank(M))$ samples per estimation, leading to a desirable total sample complexity of $O(\log \rank(M) \cdot \log \log \rank(M))$ for the chain construction.

Unfortunately, this reasoning hits a major obstacle. 
The \MinimalLink procedure iteratively defines sets $A_0, A_1, \ldots$ until they stabilize, which could potentially take up to $\Omega(\rank(M))$ iterations per link.
Even $O(\log \rank(M))$ iterations would be too much, as it will require \MinimalLink to estimate $\Omega(\log \rank(M))$ quantities for each of the $O(\log \rank(M))$ links, leading to a suboptimal total sample complexity of $\Omega(\log^2 \rank(M))$.

One of our key technical insights is that the stabilization problem can be circumvented by truncating the iterative construction process of the links after a random number $\overline{h}$ of steps drawn from a carefully designed distribution over $\{1, 2, \ldots, O(\log\log \rank(M))\}$. 
By deliberately stopping the process after $\overline{h}$ iterations, rather then letting it continue until stabilization, we gain explicit control over the number of estimations used per link, which in turn allows us to manage the total sample complexity. More specifically, 
%
recall that the spanning chain consists of $O(\log \rank(M))$ links. 
For each link, our procedure needs $O(\log \log \rank(M))$ estimations by design, 
and following our earlier reasoning, each of these estimations can be performed with high accuracy using $O(\log \log \rank(M))$ samples. 
The total sample complexity of our approach for constructing the spanning chain is the product of these terms, i.e.,
\[
\underbrace{O(\log \rank(M))}_{\text{links in chain}} \times \underbrace{O(\log \log \rank(M))}_{\text{estimations per link}} \times \underbrace{O(\log \log \rank(M))}_{\text{samples per estimation}} = O(\log \rank(M) \cdot \log^2 \log \rank(M)).
\]
Notice that this matches the near-optimal sample complexity guaranteed by \Cref{thm:main}.
We note, however, that a key challenge is to ensure that the randomized truncation does not significantly affect the expected $\mathcal{C}$-freeness of the elements. This challenge is the reason why the distribution of $\overline{h}$ has to be carefully chosen.

We now proceed to formally define the procedure for constructing a single link (\Cref{sec:singleLinkconstruction}), which is then used in \Cref{sec:chainConstruction} to sample from a distribution over spanning chains. This sampling procedure implements the oracle from Theorem~\ref{thm:chain-to-ocrs}, and thus, yields our OCRS (\Cref{thm:main}).

\subsection{Procedure for Constructing a Single Link}
\label{sec:singleLinkconstruction}

\newcommand{\algsinglelink}[1][]{\hyperref[algo:buildSingleLink]{\textup{\textsc{SingleOcrsLink}\ensuremath{#1}}}\xspace}

The input to the \algsinglelink[] algorithm is a matroid $M=(N,\mathcal{I})$, a vector $x\in P_M$, and parameters $\rho, \tau, \varepsilon$. Here, $\tau$ is the spanning probability threshold and $\varepsilon > 0$ is an error parameter.%
\footnote{In \algsinglelink[], as well as later on, we index the span and rank function, $\spn$ and $r$, respectively, by the corresponding matroid when we want to be explicit about the underlying matroid.}

\begin{center}
\begin{minipage}{0.95\textwidth}
\begin{mdframed}[hidealllines=true, backgroundcolor=gray!20]
  $\SingleOcrsLink(M, x, \rho, \tau,\varepsilon)$: 

\begin{algorithm2e}[H]  \label{algo:buildSingleLink}
  Let $\eta = \left\lceil 1 + \log_{1+\varepsilon} \frac{\ln \rho}{\varepsilon^3} \right\rceil$, and choose
  $\overline{h}$ at random from $\{1,\dots,  \eta \}$ according to the distribution defined by
  \begin{equation*}
    \Pr[\overline{h}=1] =  \frac{1}{(1+\varepsilon)^{\eta - 1}} \quad \mbox{and} \quad \Pr[\overline{h}=h] = \varepsilon \cdot \Pr[\overline{h} < h] \text{ for } h \in \{2, \dots, \eta\}.
 \end{equation*}
 \label{algline:chooseH}

  Let $A_0 = \emptyset$.

  Let $q = \left\lceil \frac{6}{\tau \varepsilon^2} \ln\left(\frac{\ln \rho}{\varepsilon}\right) \right\rceil$. \label{algline:chooseQ}

  \For{$h \leftarrow 1$ \KwTo $\overline{h}$\label{algline:loopInLink}}{
    Let $S_1, S_2, \ldots, S_q$ be $q$ independent samples from $\mathcal{D}(x)$.\label{algline:getSamples}
   %

      $A_h =  \{e\in N : \widehat{\Pr}_{S\sim (S_1,\ldots, S_q)} [e\in \spn_M(A_{h-1} \cup S)] > \tau\}$.
      \label{algline:defineAh}
    }

  \Return $A_{\overline{h}}$. 
\end{algorithm2e}
\end{mdframed}
\end{minipage}
\end{center}
The \algsinglelink[] procedure is used iteratively by \algbuildocrs[] (see \Cref{sec:chainConstruction}) to construct the spanning chain. In each iteration, it is called on a matroid obtained by restricting the matroid to the previous link of the chain. However, the parameter $\rho$ is always set to be roughly the rank of the \emph{original} matroid, and it remains fixed across all calls. In Step~\ref{algline:defineAh} of the algorithm, we use $\widehat{\Pr}[\cdot]$ to denote the probability of an event, estimated from $q$ samples $S_1, \dots, S_q$ of $\mathcal{D}(x)$.
Formally, for any event $\mathcal{E}$ depending on a set $S$, we have $\widehat{\Pr}_{S\sim (S_1,\dots,S_q)}[\mathcal{E}] = |\{p\in [q] \colon \text{$\mathcal{E}$ is satisfied when $S = S_p$}\}|/q$.

To see that Step~\ref{algline:chooseH} defines a valid probability distribution over $\{1, \ldots, \eta\}$, we note that its definition implies that the equality $\Pr[\overline{h} \leq h] = (1+\varepsilon) \Pr[\overline{h}\leq h-1]$ holds for every $h\in \{2, \ldots, \eta\}$. Thus,
\[
	\sum_{h = 1}^\eta \Pr[\overline{h} = h]
	=
	\Pr[\overline{h} \leq \eta] = (1+\varepsilon)^{\eta-1} \Pr[\overline{h}\leq 1] = (1+\varepsilon)^{\eta-1} \cdot \frac{1}{(1+\varepsilon)^{\eta-1}} = 1.
\]
The key properties of this distribution, which motivate its use, are that the probability of terminating after a single iteration is small ($\Pr[\overline{h}=1] \leq \frac{\varepsilon^3}{\ln \rho}$), and for every larger number $h \in \{2, \dots, \eta\}$ of iterations, the probability of stopping after $h$ iterations is small compared to the probability of stopping in an earlier iteration (i.e., $\Pr[\overline{h}=h] = \varepsilon \cdot \Pr[\overline{h}< h]$). To analyze the effect of the randomized truncation after $\bar{h}$, the concepts of $(A, \tau)$-good and $(A, \tau)$-bad elements are helpful.
\begin{definition}[$(A,\tau)$-bad, $(A,\tau)$-good]
  Consider a matroid $M=(N,\mathcal{I})$, $x\in P_M$, a number $\tau\in [0,1]$, and a set $A\subseteq N$.
  For an element $e\in N$, we say that 
  \begin{itemize}
    \item $e$ is $(A,\tau)$-bad (with respect to $M$ and $x$) if $e\not\in A$ and $\Pr[e \in \spn(A \cup R(x))] > \tau $, and
    \item $e$ is $(A,\tau)$-good (with respect to $M$ and $x$) if $e\not\in A$ and $\Pr[e \in \spn(A \cup R(x))] \leq \tau $.
  \end{itemize}
	Notice that elements of $A$ itself are neither $(A,\tau)$-bad nor $(A,\tau)$-good.
\end{definition}

To build intuition, we analyze an idealized setting in which the estimates $\widehat{\Pr}[\cdot]$ used by \algsinglelink are identical to the true probabilities they estimate. In this case, the sequence of sets $A_0, A_1, \ldots$ becomes deterministic, and the only source of randomness in the output $A=A_{\overline{h}}$ is from the stopping time $\overline{h}$.
Consider an arbitrary element $e \in N$. By the construction rule, if $e$ is $(A_{h-1}, \tau)$-bad, it is immediately included in the set $A_h$ (assuming $h \geq \bar{h}$). Thus, since an element must be outside a set to be considered bad with respect to it, $e$ can be $(A_{h-1}, \tau)$-bad for at most one index $h$. Let this unique index be $h_e$.
By the above discussion, $e$ is $(A, \tau)$-bad with respect to the final output $A=A_{\overline{h}}$ of \algsinglelink only if this procedure stops after iteration $h_e$. The probability of this event depends on the value of $h_e$.
\begin{itemize}
    \item \textbf{Case 1: $h_e > 2$.} The ``bad'' outcome for $e$ occurs if $\overline{h} = h_e - 1$. The ``good'' outcomes, where $e$ is left out of $A$ but is not bad, occur whenever $\overline{h} < h_e-1$. Since $h_e-1 \ge 2$, the distribution of $\overline{h}$ obeys the relation $\Pr[\overline{h}=h_e-1] = \varepsilon \cdot \Pr[\overline{h}<h_e-1]$, which directly translates to $\Pr[\text{$e$ is $(A, \tau)$-bad}] = \varepsilon \cdot \Pr[\text{$e$ is $(A, \tau)$-good}]$.

    \item \textbf{Case 2: $h_e = 2$.} The ``bad'' outcome occurs in this case if $\overline{h} = 1$. It never holds that $\overline{h} < h_e - 1 = 1$, so the probability of a ``good'' outcome is 0. Thus, in this special case, a multiplicative relation does not hold between the probabilities of ``good'' and ``bad'' outcomes. Instead, we rely on the fact that the distribution was designed to make the probability of the bad outcome of this case small in absolute terms ($\Pr[\overline{h}=1] \le \frac{\varepsilon^3}{\ln\rho}$).

    \item \textbf{(No Bad Outcome)} If $h_e=1$, then $e$ is in every set $A_h$, and thus, is never bad with respect to the output set $A$. Otherwise, if no $h_e$ value exists, $e$ is always good when it is not part of the output set.
\end{itemize}
Combining these cases, the probability that an element $e$ is $(A, \tau)$-bad is either zero, a small additive quantity (if $h_e=2$), or $\varepsilon$ times the probability that it is $(A, \tau)$-good (if $h_e > 2$). This provides the precise intuition for the general bound in the proposition below, which has the form $\Pr[\text{bad}] \leq \varepsilon \cdot \Pr[\text{good}] + \text{small-term}$. The fact that an element is likely to be $(A, \tau)$-good is desirable, as this property contributes towards a $\tau$-balancedness for the final chain.

The above informal argument is made rigorous in Section~\ref{sec:inlink_loss}, where we account for estimation errors to prove the following proposition. (The need to handle these sample errors is the reason that this proposition uses $(1-\varepsilon)\tau$, rather than $\tau$, as the spanning probability threshold in the invocation of \algsinglelink[].)
\begin{restatable}{proposition}{propInLinkLose} \label{prop:inLinkLose}
  Consider a matroid $M = (N, \mathcal{I})$, a vector $x\in P_M$, and numbers $\rho\in \mathbb{Z}_{\geq 3}$, $\tau\in (0,1]$, and $\varepsilon \in (0,\tau]$.
  Let $A$ be the (random) output of \algsinglelink{$(M,x, \rho,(1-\varepsilon)\tau, \varepsilon)$}.
  Then,
  \begin{equation*}
    \Pr[\text{$e$ is $(A,\tau)$-bad}] \leq \varepsilon \cdot \Pr[\text{$e$ is $(A,\tau)$-good}] + \frac{2\varepsilon^3}{\ln \rho}
    \quad\forall e\in N,
  \end{equation*}
  where being $(A,\tau)$-bad or $(A,\tau)$-good is with respect to the matroid $M$ and the vector $x$.
\end{restatable}

\subsection{Constructing the Spanning Chain}
\label{sec:chainConstruction}

The \algbuildocrs[] procedure constructs a full spanning chain by iteratively applying the \algsinglelink[] subroutine. The length $\zeta$ of the chain is determined by the rank $\rho$ of the matroid\footnote{If the rank of the matroid is a small constant, $\rho$ is chosen to be a bit larger than it for technical reasons.} and the error parameter $\varepsilon$. After determining $\zeta$, the procedure initializes the first link of the chain to be the entire ground set $C_0 = N$, and the last link of the chain to be the empty set $C_{\zeta+1} = \emptyset$. 
To construct the remaining links of the spanning chain, \algbuildocrs[] iterates from $i=1$ to $\zeta$. In each iteration $i$, it generates the link $C_i$ by invoking \algsinglelink[]. Crucially, the input to the last subroutine is $M\vert_{C_{i-1}}$, i.e., the matroid $M$ restricted to the elements of the previous link $C_{i-1}$. This restriction naturally ensures that the resulting sequence of links is nested, i.e., $C_0 \supseteq C_1 \supseteq \dots \supseteq C_{\zeta+1}$. 

\begin{center}
\begin{minipage}{0.95\textwidth}
\begin{mdframed}[hidealllines=true, backgroundcolor=gray!20]
  \textsc{OcrsChain}$(M, x, \tau, \varepsilon)$:\\
\begin{algorithm2e}[H] \label{algo:buildOCRS}
  Let $\rho = \max\{\rank(M), 3\}$.

  Let $\zeta = \big\lceil \frac{1}{\varepsilon}\ln\left(\frac{\rho}{\varepsilon}\right) \big\rceil$.


  Let $C_0 = N$ and $C_{\zeta+1} = \emptyset$.


  \For{$i \leftarrow 1$ \KwTo $\zeta$}{
    Let $C_{i}$ be the output of \algsinglelink{$(M\vert_{C_{i-1}},  x\vert_{C_{i-1}},\rho, (1-\varepsilon)\tau, \varepsilon)$}.
  }

  \Return $(C_0, C_1, \dots, C_{\zeta+1})$.
\end{algorithm2e}
\end{mdframed}
\end{minipage}
\end{center}

The following is a corollary of Proposition~\ref{prop:inLinkLose}. It shows that elements that are unlikely to be part of the last link of the spanning chain constructed by \algbuildocrs[] enjoy a good $\mathcal{C}$-freeness.



\begin{restatable}{corollary}{corFreenessLikely}\label{cor:freenessLikely}
  Let $M$ be a matroid, $\varepsilon \in (0,\sfrac{1}{4})$, $\lambda \in (0,1-4\varepsilon]$, and $x\in \lambda \cdot P_M$.
  Let $\mathcal{C}=(C_0,\dots,C_{\zeta+1})$ be the (random) output of \algbuildocrs{$(M,\mathcal{D}(x),\lambda+4\varepsilon, \varepsilon)$}.
  Then, for every element $e \in N$,
	\[
		\Pr[\text{$\mathcal{C}$-freeness of $e$ is at least $1-\lambda-4\varepsilon$} \mid e \not \in C_{\zeta}]
		\geq
		1 - \varepsilon - \frac{2\varepsilon}{\Pr[e\not\in C_\zeta]}.
	\]
\end{restatable}
\begin{proof}
  Let $C_j\in \mathcal{C}$ be the set with highest index that contains $e$.
  In other words, $C_j$ is a smallest set containing $e$ in the sequence $C_0,\dots, C_{\zeta}$; and if this smallest set appears multiple times, we take the one with the highest index. Notice that conditioning on $e \not \in C_{\zeta}$ is equivalent to conditioning on the event $j\in \{0,\dots, \zeta-1\}$.

	Observe that $x|_{C_i}$ is a vector in the matroid polytope of $M\vert_{C_{i}}$ because, for every set $S \subseteq C_i$, $\sum_{e \in S} x_e \leq r_M(S) = r_{M|_{C_i}}(S)$.
	Thus, it makes sense to talk about an element being good or bad with respect to the matroid $M\vert_{C_{i}}$ and the vector $x\vert_{C_{i}}$.  To simplify the terminology, we say, for $i\in [\zeta]$, that $e$ is \emph{$i$-good} if $e$ is $(C_{i+1}, \lambda+4\varepsilon)$-good with respect to the matroid $M\vert_{C_{i}}$ and the vector $x\vert_{C_{i}}$, and that $e$ is \emph{$i$-bad} if $e$ is $(C_{i+1}, \lambda+4\varepsilon)$-bad with respect to this matroid and vector.
  Then, the $\mathcal{C}$-freeness of $e$ is at least $1-\lambda-4\varepsilon$, if and only if $e$ is $j$-good. Thus,
  \begin{align*}
		\Pr[\mspace{100mu}&\mspace{-100mu}\text{$\mathcal{C}$-freeness of $e$ is at least $1-\lambda-4\varepsilon$} \mid e \not \in C_{\zeta}]\\
		&= \Pr[\text{$e$ is $j$-good} \mid e \not \in C_{\zeta}] \\&= 1-\Pr[\text{$e$ is $j$-bad} \mid e \not \in C_{\zeta}]\\
		&= 1 - \frac{1}{\Pr[e\not\in C_{\zeta}]}\sum_{i = 0}^{\zeta - 1} \Pr[j = i] \cdot \Pr[\text{$e$ is $i$-bad} \mid j = i]\\
		&\geq 1 - \frac{1}{\Pr[e\not\in C_\zeta]}\sum_{i=0}^{\zeta - 1} \left(\Pr[j = i] \cdot \varepsilon \cdot \Pr[\text{$e$ is $i$-good} \mid j = i] + \Pr[j \geq i] \cdot \frac{2\varepsilon^3}{\ln \rho}\right) \\
     &= 1 - \varepsilon \cdot \Pr[\text{$e$ is $j$-good} \mid e \not \in C_{\zeta}] - \frac{1}{\Pr[e\not\in C_\zeta]}\sum_{i=0}^{\zeta - 1} \Pr[j \geq i] \cdot \frac{2\varepsilon^3}{\ln \rho}\\
		&\geq
		1 - \varepsilon - \frac{1}{\Pr[e\not\in C_\zeta]}\cdot \zeta \cdot \frac{2\varepsilon^3}{\ln \rho}
    ,
  \end{align*}
	where the second equality holds since the fact that $e\not\in C_{j+1}$ implies that $e$ must be either $j$-good or $j$-bad. To see why the first inequality holds as well, observe that $e$ can be either $i$-good or $i$-bad only if $j = i$. Furthermore, the event $j \geq i$ is independent of the result of the random construction of $C_{i + 1}$. Thus, by \Cref{prop:inLinkLose},
	\begin{multline*}
		\Pr[\text{$e$ is $i$-bad} \mid j = i]
		=
		\frac{\Pr[\text{$e$ is $i$-bad} \mid j \geq i]}{\Pr[j = i \mid j \geq i]}\\
		\leq
		\frac{\varepsilon \cdot \Pr[\text{$e$ is $i$-good} \mid j \geq i] + \frac{2\varepsilon^3}{\ln \rho}}{\Pr[j = i \mid j \geq i]}
		=
		\varepsilon \cdot \Pr[\text{$e$ is $i$-good} \mid j = i] + \frac{\Pr[j \geq i]}{\Pr[j = i]} \cdot \frac{2\varepsilon^3}{\ln \rho}
		.
	\end{multline*}
%
	The corollary now follows since $\rho \geq 3$ guarantees that
	\[
		\zeta \cdot \frac{\varepsilon^3}{\ln \rho}
		\leq
		\Big(1 + \frac{1}{\varepsilon}\ln\Big(\frac{\rho}{\varepsilon}\Big)\Big) \cdot \frac{\varepsilon^3}{\ln \rho}
		\leq
		\varepsilon^3 + \varepsilon^2(1 + \ln \varepsilon^{-1})
		\leq
		\varepsilon^3 + \varepsilon
		\leq
		2\varepsilon
		.
		\tag*{\qedhere}
	\]
\end{proof}

The preceding corollary establishes that, with high probability, elements achieve a high level of $\mathcal{C}$-freeness whenever they are likely to be outside the terminal set $C_{\zeta}$. For this to be useful, we need to show that no element is likely to be in $C_{\zeta}$, which we do by showing that $C_\zeta$ is empty with high probability.
The proof of \cite{FeldmanSZ21} for the known $x$ setting can easily be adapted to show that this holds in an idealized world in which the estimates $\widehat{\Pr}[\cdot]$ are all perfect.
Unfortunately, however, our sampling-based approach introduces estimation errors that require a careful analysis.
In particular, because the construction of a link by \algsinglelink[] is done iteratively through the construction of multiple sets $A_h$, we must bound the accumulated effect of the sampling errors in all the iterations.
To this end, to analyze the sequence $A_1, A_2, \dots$ of sets, we introduce in Section~\ref{sec:progress} another sequence $B_1,B_2,\dots$ of sets that dominates the former one in the sense that $A_h \subseteq B_h$ for every $h$, and is easier to analyze.
We provide more details in Section~\ref{sec:progress}, where we prove the following proposition.

\begin{restatable}{proposition}{propProgress} \label{prop:progress}
  Let $M$ be a matroid, $\lambda \in (0,1)$, $x\in \lambda \cdot P_M$, $\tau\in (\lambda,1]$, $\rho\in \mathbb{Z}_{\geq 3}$, and $\varepsilon \in (0,\sfrac{1}{20}]$.
  Let $A$ be the output of \algsinglelink{$(M,\rho,\mathcal{D}(x),(1-\varepsilon)\tau, \varepsilon)$}.
  Then
  \begin{equation*}
    \E[r(A)] \leq \left(1 + \lambda - (1-3\varepsilon)\tau \right)\cdot \rank(M).
  \end{equation*}
\end{restatable}

\begin{restatable}{corollary}{corChainIsSpanning}\label{cor:chainIsSpanning}
  Let $M$ be a matroid, $\varepsilon \in (0,\sfrac{1}{20}]$, $\lambda \in (0,1-4\varepsilon]$, and $x\in \lambda \cdot P_M$.
  Let $\mathcal{C}=(C_0,\dots,C_{\zeta+1})$ be the (random) output of \algbuildocrs{$(M,\mathcal{D}(x),\lambda+4\varepsilon,\varepsilon)$}.
  Then, with probability at least $1-\varepsilon$, we have $C_{\zeta}=\emptyset$.
\end{restatable}
\begin{proof}
%
	Consider an arbitrary integer $1 \leq i \leq \zeta$, and observe that $x|_{C_{i - 1}} \in \lambda \cdot P_{M|_{C_{i - 1}}}$ because, for every $S \subseteq C_{i - 1}$, it holds that $\sum_{e \in S} x_e \leq \lambda \cdot r_M(S) = \lambda \cdot r_{M|_{C_{i - 1}}}(S)$. Thus, by \Cref{prop:progress},
	\begin{align*}
	\E[r(C_i) \mid C_{i-1}] &\leq \left(1 + \lambda - (1-3\varepsilon)(\lambda + 4\varepsilon)\right) \cdot \rank(M|_{C_{i-1}})\\
                              &=\left(1 + 3\varepsilon\lambda - 4\varepsilon + 12\varepsilon^2\right) \cdot r(C_{i-1})\\
                              &\leq (1 - \varepsilon) \cdot r(C_{i-1}),
	\end{align*}
	where the last inequality uses $\lambda \leq 1 - 4\varepsilon$. By the law of total expectation, the last inequality implies $\E[r(C_i)] \leq (1 - \varepsilon) \cdot \E[r(C_{i-1})]$, and combining this inequality for all values of $i$ yields
%
  \begin{equation*}
    \E[r(C_{\zeta})] \leq (1-\varepsilon)^{\zeta} \cdot \E[r(C_0)]
			=
			(1-\varepsilon)^{\zeta} \cdot r(N)
      \leq e^{-\varepsilon \cdot \zeta} \cdot \rho
      \leq \varepsilon,
  \end{equation*}
  where we use in the penultimate inequality that $\rho$ is set to be at least $r(N)$ by \algbuildocrs{}.
By Markov's inequality, the last bound on $\E[r(C_{\zeta})]$ implies that
  \begin{equation*}
    \Pr[C_{\zeta} = \emptyset] = \Pr[r(C_{\zeta}) = 0] = 1 - \Pr[r(C_{\zeta}) \geq 1] \geq 1 - \varepsilon.
		\tag*{\qedhere}
  \end{equation*}
\end{proof}

Combining the above results, we get the next theorem, which summarizes the properties of the spanning chain produced by \algbuildocrs[].
\begin{theorem} \label{thm:balance}
  Let $M$ be a matroid, $\varepsilon \in (0,\sfrac{1}{20}]$, $\lambda \in (0,1-4\varepsilon]$, and $x\in \lambda \cdot P_M$.
  Let $\mathcal{C}=(C_0,\dots,C_{\zeta+1})$ be the (random) output of \algbuildocrs{$(M,x,\lambda+4\varepsilon, \varepsilon)$}.
  Then, the distribution of $\mathcal{C}$ is a $(1-\lambda - 8\varepsilon)$-balanced spanning chain distribution for $x$ with respect to $M$.
\end{theorem}
\begin{proof}
The expected $\mathcal{C}$-freeness of each element $e \in N$ can be lower bounded by
\begin{align*}
	\Pr[e \not\in\ &C_{\zeta}] \cdot \Pr[\text{$\mathcal{C}$-freeness of $e$ is at least $1 - \lambda - 4\varepsilon$} \mid e \not \in C_{\zeta}] \cdot (1 - \lambda - 4\varepsilon)\\
                 &\geq \left(\Pr[e\not\in C_\zeta] (1-\varepsilon) -2\varepsilon\right)(1 - \lambda - 4\varepsilon)\\
                 &\geq \left(\Pr[C_\zeta = \emptyset] (1-\varepsilon) -2\varepsilon\right)(1 - \lambda - 4\varepsilon)\\
                 &\geq \left((1-\varepsilon)^2 - 2\varepsilon\right)(1 - \lambda - 4\varepsilon)\\
                 &\geq (1-4\varepsilon)(1 - \lambda - 4\varepsilon)\\
                 &\geq 1 - \lambda - 8\varepsilon
	,
\end{align*}
where the first inequality holds by \Cref{cor:freenessLikely} and the third one by \Cref{cor:chainIsSpanning}.
\end{proof}

Since sample access to $R(\lambda x)$ can be simulated via sample access to $R(x)$, as was mentioned in the beginning of this section, \Cref{thm:balance} can be used to implement the oracle required by \Cref{thm:chain-to-ocrs}. For the choice of $\lambda = 1/2$, this proves \Cref{thm:main}.\footnote{Technically, we get in this way a weaker version of Theorem~\ref{thm:main} in which $\varepsilon$ is replaced with $O(\varepsilon)$. However, this can be fixed by simply scaling down $\varepsilon$ by an appropriate constant.}

\section{Guaranteeing Progress in Link} \label{sec:progress}

In this section, we prove \Cref{prop:progress}, which we repeat below for convenience.

\propProgress*

Note that the procedure \algsinglelink{$(M,n,\mathcal{D}(x), (1-\varepsilon)\tau, \varepsilon)$} has two sources of randomness.
The first one is the choice of $\overline{h}$, and the second one is the samples taken in line~\ref{algline:getSamples} of the procedure.
To keep these separate, we would like the set $A_i$ to be defined, for every integer $0 \leq i \leq  \eta$, regardless of the value of $\overline{h}$. This can be achieved by pretending (just for analysis purposes) that the upper bound of the loop on line~\ref{algline:loopInLink} of \algsinglelink is $\eta$ rather than $\overline{h}$. 
This does not affect the output set $A \coloneqq A_{\overline{h}}$ of this procedure.
As a consequence of the fact that the sets $A_0, A_1, \dots, A_{\eta}$ are all always defined now, we get the following observation.

\begin{observation} \label{obs:rank_max_at_eta}
It holds that $\E[r(A)] \leq \E[r(A_{\eta})]$.
\end{observation}
\begin{proof}
The observation holds since $A$ is one of the sets $A_1, A_2, \dots, A_{\eta}$, and these sets obey $A_1 \subseteq A_2 \subseteq \dots \subseteq A_{\eta}$.
\end{proof}

In light of \Cref{obs:rank_max_at_eta}, to prove \Cref{prop:progress}, it suffices to upper bound $\E[r(A_{\eta})]$, which is our goal in the rest of the section. A natural way to get such an upper bound would be to upper bound how quickly $\E[r(A_i)]$ increases when viewed as a function of $i$. However, we do not know of an easy way to do that directly. Instead, we define a new series of sets $B_0, B_1, \dotsc, B_\eta$, such that $A_i \subseteq B_i$ for every integer $0 \leq i \leq \eta$, and then upper bound the rate in which $\E[r(B_i)]$ increases as a function of $i$. To define this new series, we need the following tool (\Cref{lem:rankReductionInOverlap}), a formal proof of which can be found in \Cref{ssc:rankReductionInOverlap}. In this tool, and in the rest of this section, given two sets $S,T \subseteq M$, we use the standard shorthand $r(S \mid T) \coloneq r(S \cup T) - r(T)$. We also refer to this expression as the \emph{marginal contribution of the set $S$ with respect to the set $T$}.

\begin{restatable}{lemma}{lemRankReductionInOverlap}\label{lem:rankReductionInOverlap}
  For every set $B\subseteq N$ and value $\alpha \in [0,1)$, let us define
  \begin{equation*}
    T_\alpha(B) \coloneqq \argmax_{\substack{\bar{T} \subseteq N:\\B \subseteq \bar{T}}} \Big\{ r(\bar{T}\ |\ B) - \frac{1}{1-\alpha}\E[r(\bar{T}\ |\ B\cup R(x))] \Big\}.
  \end{equation*}
  Then, $B \subseteq T_\alpha(B)$. Furthermore, 
	\begin{itemize}
		\item $\E[r(T_\alpha(B) \ |\ B\cup R(x))] \leq (1-\alpha)\cdot r(T_\alpha(B)\ |\ B)$; and
    \item  $\E[r(Q\cap \spn(T_\alpha(B)\cup R(x)) \ |\ T_\alpha(B) )] \leq \alpha\cdot r(Q \ |\ T_\alpha(B))$
      for every set $Q \subseteq N \setminus T_\alpha(B)$.
	\end{itemize}
\end{restatable}

Intuitively, $T_\alpha(B)$ is an extension of $B$ obtained by adding to $B$ sets whose expected marginal contribution with respect to $B \cup R(x)$ is much smaller than their expected marginal contribution with respect to $B$ itself. This naturally implies that the expected marginal contribution of the entire set $T_\alpha(B)$ with respect to $B\cup R(x)$ is much smaller than the marginal contribution of $T_\alpha(B)$ with respect to $B$, which is the first bullet point of \Cref{lem:rankReductionInOverlap}. Another consequence of the definition of $T_\alpha(B)$ is that the span of any sample, which has the same distribution as $R(x)$, is unlikely to substantially overlap (in terms of span) with any set $Q$ that is disjoint from $T_\alpha(B)$. This is formally stated by the second bullet point of \Cref{lem:rankReductionInOverlap}, which uses the marginal contribution with respect to $T_\alpha(B)$ as a way to measure the ``overlap.''

The last property is very helpful for bounding sample errors when computing sets of the type
\begin{equation*}
  \{e\in N \colon \widehat{\Pr}_{S \sim (S_1, S_2, \dotsc, S_q)} [e\in \spn(T_\alpha(B) \cup S)]>\tau\}
\end{equation*}
for some set $B$. Notice that in \algsinglelink we compute similarly looking sets, but these sets have $A_{h-1}$ in their definitions instead of $T_\alpha(B)$.
This discrepancy is, intuitively, the reason that we have to define the above-mentioned auxiliary sequence $B_0,B_1,\dots,B_\eta$. To formally define this sequence, let us denote by $S_p^h$ the sample set $S_p$ in iteration number $h$ of the loop of the procedure \algsinglelink. Recall that $A_0 = \emptyset$, and for every $h \in [\eta]$,
\begin{equation*}
  A_h = \{e\in N \colon \widehat{\Pr}_{S \sim (S_1^h, S_2^h, \dotsc, S_q^h)} [e\in \spn(A_{h-1} \cup S)] > (1-\varepsilon)\tau\}. 
\end{equation*}
Similarly, we define $B_0 \coloneq \emptyset$, and for every $h \in [\eta]$,
\begin{equation}\label{eq:BhDefinition}
	B_h \coloneq \{e\in N \colon \widehat{\Pr}_{S \sim (S_1^h, S_2^h, \dotsc, S_q^h)} [e \in \spn(T_\alpha(B_{h - 1}) \cup S)] > (1-\varepsilon)\tau\},
\end{equation}
where $\alpha$ is chosen to be 
\begin{equation}
\alpha = \tau(1 - 2\varepsilon)\,.
\label{eq:definitionalpha}
\end{equation}

We now observe that, as was promised above, the sequence $B_0,\dots, B_\eta$ dominates the sequence $A_0,\dots, A_\eta$ in the sense that $A_h \subseteq B_h$ for every $h \in [\eta]$. 
Notice that this is useful since it implies that we can upper bound $\E[r(A_\eta)]$ by upper bounding $\E[r(B_\eta)]$.

\begin{lemma}\label{lem:AhSubseteqBh}
For every $h\in \{0,\dots, \eta\}$, $A_h \subseteq B_h$.
\end{lemma}
\begin{proof}
  We prove the statement by induction.
  The base case $h=0$ is trivial since $A_0=B_0=\emptyset$.
  For the induction step, the induction hypothesis is that $A_{h-1} \subseteq B_{h-1}$ for some $h\in \{1,\dots, \eta\}$, and we need to show $A_h \subseteq B_h$.
  Since \Cref{lem:rankReductionInOverlap} guarantees that $B_{h - 1} \subseteq T(B_{h - 1})$, by the monotonicity of the span function, we have $\spn(A_{h-1}\cup S) \subseteq \spn(B_{h-1}\cup S) \subseteq \spn(T(B_{h - 1})\cup S)$.
  Recall now that $A_h$ and $B_h$ are defined identically, with the only difference being that the expression $\spn(A_{h-1}\cup S)$ in the definition of $A_h$ is replaced with $\spn(T(B_{h - 1}) \cup S)$ in the definition of $B_h$. Since we have already seen that $\spn(A_{h-1}\cup S) \subseteq \spn(T(B_{h - 1})\cup S)$, this implies $A_h \subseteq B_h$, as desired.
\end{proof}

To bound the expected rank of the sets in $B_0,\dots, B_\eta$, it is helpful to think of the construction of $B_h$ from $B_{h-1}$ as happening in two steps. 
In a first step, we compute the set $T_\alpha(B_{h - 1})$, which is a superset of $B_{h - 1}$, and in a second step we compute $B_h$ from $T_\alpha(B_{h-1})$ as shown in \Cref{eq:BhDefinition}.
The following lemma shows that in expectation the rank $r(B_h)$ is not much larger than the rank of $T_\alpha(B_{h - 1})$.
Hence, the second step of this two-step process does not add much to the expected rank of $B_h$.

\begin{restatable}{lemma}{lemBNotGrowingMuch} \label{lem:BNotGrowingMuch}
  For every $h \in [\eta]$, $\E[r(B_h\ |\ T_\alpha(B_{h - 1}))] \leq e^{-\frac{1}{2} \varepsilon^2 \tau q} \cdot r(N) \leq \left(\frac{\varepsilon}{\ln \rho}\right)^3\cdot r(N)$.
\end{restatable}

The second inequality of the lemma follows immediately from the definition of $q$.
Proving the first inequality is much more involved, and is quite technically challenging. Thus, we defer the formal proof of this inequality to \Cref{sec:BNotGrowingMuch},
and in the following lines, we explain why it intuitively follows from the properties of $T_\alpha(B_{h-1})$ given by \Cref{lem:rankReductionInOverlap}.
Notice that $B_h$ includes elements that are likely to be spanned by $T_\alpha(B_{h - 1}) \cup S$, and $T_\alpha(B_{h - 1})$ itself already includes every element that is likely to be spanned by $T_\alpha(B_{h - 1}) \cup R(x)$. Since $S$ is a sample from the distribution $\mathcal{D}(x)$ of $R(x)$, up to sampling errors, we expect all the elements of $B_h$ to already be in $T_\alpha(B_{h - 1})$. Thus, in a world of perfect samples, $r(B_h\ |\ T_\alpha(B_{h - 1}))$ would be exactly $0$. 
Bounding the deviation from this ideal due to sample errors is the main challenge of the proof.
To this end, the following is a helpful viewpoint.
When checking which elements $e\in N$ fulfill
\[
  \widehat{\Pr}_{S \sim (S_1^h, S_2^h, \dots, S_q^h)} [e \in \spn(T_\alpha(B_{h - 1}) \cup S)] > (1-\varepsilon)\tau,
\]
we consider the samples $S_1^h, S_2^h, \dots, S_q^h$ one by one, starting from $S_1^h$.
Assume that we have considered the samples $S_1^h, S_2^h, \dots$ up to $S_p^h$ so far.
Then we keep track, for each integer $i$, of how big the rank is of the elements
\[
  Z_{p,i} \coloneqq \{e\in N\setminus T_{\alpha}(B_{h-1}) \colon |\{j\in [p]\colon e\in \spn(S^h_j\cup T_\alpha(B_{h-1}))\}|\geq i\}.
\]
In words, $Z_{p,i}$ is the set of elements outside of $T_\alpha(B_{h-1})$ that are spanned by at least $i$ many of the samples $S_1^h, S_2^h, \dotsc, S_p^h$ together with $T_\alpha(B_{h-1})$.
The second bullet point of \Cref{lem:rankReductionInOverlap} guarantees that, when considering the next sample $S_{j+1}^h$, the overlap of $\spn(T_\alpha(B_{h-1})\cup S_j^h)$ with $Z_{p,i}$ (this set takes the role of $Q$ in \Cref{lem:BNotGrowingMuch}) is small when measured in terms of the marginal contribution to the rank (with respect to $T_\alpha(B_{h - 1})$).
This shows that, in some sense, $r(Z_{p,i} \ |\ T_\alpha(B_{h-1}))$ behaves like a binomial random variable, which allows us to bound its expectation. We then get a bound on the desired quantity $\E[r(B_h\ |\ T_\alpha(B_{h-1}))]$ because $B_h\setminus T_\alpha(B_{h-1}) \subseteq Z_{q,\lceil(1-\varepsilon)\tau q \rceil}$.%
\footnote{An equality between $B_h\setminus T_{\alpha}(B_{h-1})$ and $Z_{q,\lceil(1-\varepsilon)\tau q \rceil}$ would have held if the definition of $B_{h}$ had used a non-strict inequality.}

Using \Cref{lem:BNotGrowingMuch}, we are now ready to prove \Cref{prop:progress}.

\begin{proof}[Proof of \Cref{prop:progress}]
The rank function $r$ of a matroid is monotone and submodular. These properties imply that
\begin{align} \label{eq:B_eta_bound}
	\E[r(B_{\eta})] &\leq \E[r(B_{\eta} \cup R(x))] \\\nonumber
                                  &= \E[r(R(x))] + \sum_{h=1}^{\eta} \E[r(B_h\ |\ B_{h-1}\cup R(x)]\\\nonumber
                                  &= \E[r(R(x))] + \sum_{h=1}^{\eta}\Big(
                                    \E[r( T_\alpha(B_{h - 1}) \ |\ B_{h-1}\cup R(x))] + \E[r(B_h \ |\ T_\alpha(B_{h - 1}) \cup R(x))]
                                    \Big)\\\nonumber
                                  &\leq \E[r(R(x))] + \sum_{h=1}^{\eta}\Big(
                                    (1-\alpha)\cdot \E[r( T_\alpha(B_{h - 1})\ |\ B_{h-1})] + \E[r(B_h \ |\  T_\alpha(B_{h - 1}))]
                                    \Big)\\\nonumber
                                  &= \E[r(R(x))] + (1-\alpha)\cdot \sum_{h=1}^{\eta} \E[r(B_h\ |\ B_{h-1})]
                                  + \alpha\cdot \sum_{h=1}^{\eta} \E[r(B_h \ |\ T_\alpha(B_{h - 1}))],
\end{align}
where the equalities exploit the fact that $B_{h-1} \subseteq T_\alpha(B_{h - 1}) \subseteq B_h$ for every $h \in [\eta]$, and the second inequality follows from the first bullet point of \Cref{lem:rankReductionInOverlap} (and the submodularity of $r$).

We now would like to upper bound each one of the terms on the rightmost side of \Cref{eq:B_eta_bound}. First, recall that $x\in \lambda\cdot P_M$, which implies $\E[r(R(x))] \leq \E[|R(x)|] = \|x\|_1 \leq \lambda\cdot r(N)$ because $r(U) \leq |U|$ for every set $U \subseteq N$. Second, by the monotonicity of $r$ and the fact that $r(B_0) = r(\varnothing) = 0$, we get $\sum_{h=1}^{\eta}{r(B_h \ |\ B_{h-1})} = r(B_\eta) - r(B_0) = r(B_\eta) \leq r(N)$. Finally, \Cref{lem:BNotGrowingMuch} guarantees that $\E[r(B_h \ |\ T_\alpha(B_{h - 1}))] \leq e^{-\frac{1}{2} \varepsilon^2 \tau q} \cdot r(N) \leq \left(\frac{\varepsilon}{\ln \rho}\right)^3\cdot r(N)$ for every $h\in [\eta]$. Plugging all these bounds into \Cref{eq:B_eta_bound}, we get
  \begin{align*}
    \E[r(B_{\eta})]
                                  &\leq \lambda\cdot r(N) + (1-\alpha)\cdot r(N) + \eta \alpha e^{-\frac{1}{2}\varepsilon^2 \tau q} \cdot r(N) \\
                                  &\leq \left(1 + \lambda - \alpha + \eta \alpha \bigg(\frac{\varepsilon}{\ln \rho}\bigg)^3 \right)\cdot r(N)\\
                                  &\leq \left(1 + \lambda - (1-2\varepsilon)\tau + \eta \tau \left(\frac{\varepsilon}{\ln \rho}\right)^3 \right)\cdot r(N)\\
                                  &\leq \left(1 + \lambda - (1-3\varepsilon)\tau \right)\cdot r(N),
  \end{align*}
  where the second inequality holds by the definition of $q$ (line~\ref{algline:chooseQ} of \algsinglelink[]), the third inequality follows by recalling that $\alpha = \tau(1 - 2\varepsilon) \leq \tau$,
  and the last inequality holds because $\eta = \lceil 1 + \log_{1+\varepsilon}\frac{\ln \rho}{\varepsilon^3} \rceil$, which implies
  \begin{equation*}
    \eta \leq 2 + \log_{1+\varepsilon}\frac{\ln \rho}{\varepsilon^3} = 2 + \frac{\ln\ln\rho - 3\ln \varepsilon}{\ln (1+\varepsilon)} 
    \leq 2 + \frac{2}{\varepsilon} \left(\ln\ln \rho + 3 \ln \frac{1}{\varepsilon}\right) \leq \frac{\ln \rho}{\varepsilon^2}
  \end{equation*}
  (the two last inequalities hold for $\rho\geq 3$ and $\varepsilon \leq \sfrac{1}{20}$).

	To complete the proof of the proposition, it only remains to recall that, by \Cref{obs:rank_max_at_eta} and \Cref{lem:AhSubseteqBh},
  \[
		\E[r(A)] \leq \E[r(A_{\eta})] \leq
    \E[r(B_{\eta})] \leq \left(1 + \lambda - (1-3\varepsilon)\tau \right)\cdot \rank(M). \tag*{\qedhere}
  \]
\end{proof}

\subsection{Proof of \Cref{lem:rankReductionInOverlap}} \label{ssc:rankReductionInOverlap}

In this section, we prove \Cref{lem:rankReductionInOverlap}, which we repeat here for convenience.

\lemRankReductionInOverlap*

Notice that the inclusion $B \subseteq T_\alpha(B)$ is an immediate consequence of the definition of $T_\alpha(B)$. Thus, we concentrate on proving the two bullet points of the lemma. Since $B$ itself is a possible choice for $\bar{T}$ in the definition of $T_\alpha(B)$, it must hold that
\[
  r(T_\alpha(B)\ |\ B) - \frac{1}{1-\alpha}\E[r(T_\alpha(B)\ |\ B\cup R(x))]
     \geq r(B\ |\ B) - \frac{1}{1-\alpha}\E[r(B\ |\ B\cup R(x))] = 0,
\]
and the first bullet point of \Cref{lem:rankReductionInOverlap} follows by rearranging this inequality.

Proving the second bullet point of \Cref{lem:rankReductionInOverlap} is more involved, and is our goal in the rest of this section. We start by proving the following inequality.
\begin{equation}\label{eq:overlapIsLessThanRankReduction}
  \E[r(Q\cap \spn(T_{\alpha}(B) \cup R(x)) \ |\ T_\alpha(B) )] \leq r(Q\ |\ T_\alpha(B)) - \E[r(Q \ |\ T_\alpha(B) \cup R(x))]
\end{equation}
By expanding the marginal contribution notation, we get that \Cref{eq:overlapIsLessThanRankReduction} is equivalent to
\begin{multline*}
  \E[r((Q \cup T_\alpha(B))\cap \spn(T_\alpha(B) \cup R(x)))] \\\leq r(Q \cup T_\alpha(B)) - \E[r(Q\cup T_\alpha(B) \cup R(x))] + \E[r(T_\alpha(B) \cup R(x))],
\end{multline*}
which holds because
\begin{align*}
  r(Q\cup {}&T_\alpha(B)) + \E[r(T_\alpha(B) \cup R(x))]\\
      &= r(Q\cup T_\alpha(B)) + \E[r(\spn(T_\alpha(B) \cup R(x)))]\\
      &\geq \E[r((Q\cup T_\alpha(B))\cup \spn(T_\alpha(B) \cup R(x)))] + \E[r((Q\cup T_\alpha(B))\cap \spn(T_\alpha(B)\cup R(x)))]\\
      &= \E[r(Q\cup T_\alpha(B) \cup R(x))] + \E[r((Q\cup T_\alpha(B))\cap \spn(T_\alpha(B)\cup R(x)))],
\end{align*}
where the inequality is due to submodularity of the rank function $r$, and the two equalities follow from the facts that $r(U) = r(U')$ for any $U \subseteq U' \subseteq \spn(U)$.
In particular, for any two sets $U,W \subseteq N$, we thus have $r(U\cup \spn(W)) = r(U\cup W)$, which we use in the last equality.

Having proved \Cref{eq:overlapIsLessThanRankReduction}, we now proceed by bounding the right-hand side of this inequality.
To this end, observe that, because $T_\alpha(B) \cup Q$ is one possible choice for the set $\bar{T}$ in the definition of $T_\alpha(B)$, it must hold that
\begin{multline*}
r(T_\alpha(B)\ |\ B) - \frac{1}{1-\alpha}\E[r(T_\alpha(B)\ |\ B\cup R(x))]\\ \geq r(T_\alpha(B)\cup Q \ |\ B) - \frac{1}{1-\alpha}\E[r(T_\alpha(B)\cup Q\ |\ B\cup R(x))].
\end{multline*}
By expanding the marginal contribution notation, and recalling that $B\subseteq T_\alpha(B)$, the last inequality is equivalent to
\begin{equation*}
  r(T_\alpha(B)) - \frac{1}{1-\alpha} \E[r(T_\alpha(B)\cup R(x))] \geq
  r(T_\alpha(B)\cup Q) - \frac{1}{1-\alpha} \E[r(T_\alpha(B)\cup Q\cup R(x))],
\end{equation*}
which can be rewritten as
\begin{equation*}
  r(Q \ |\ T_\alpha(B)) - \E[r(Q \ |\ T_\alpha(B)\cup R(x))] \leq \alpha \cdot r(Q \ |\ T_\alpha(B)).
\end{equation*}
The second bullet point of \Cref{lem:rankReductionInOverlap} now follows by combining this inequality with \Cref{eq:overlapIsLessThanRankReduction}.

\subsection{Proof of \Cref{lem:BNotGrowingMuch}}\label{sec:BNotGrowingMuch}

We now present the proof of \Cref{lem:BNotGrowingMuch}, which we repeat here for convinience.

\lemBNotGrowingMuch*

Fix a particular value for $h \in [\eta]$, and recall that $S_p^h$ is the $p$-th sample set used to construct $A_h$ and $B_h$.
We also recall the definition
\[
	Z_{p,i} \coloneqq \{e\in N\setminus T_{\alpha}(B_{h-1}) \colon |\{j\in [p]\colon e\in \spn(S^h_j\cup T_\alpha(B_{h - 1}))\}|\geq i\}
\]
for every integer $p \in [q]$ and $i\in \mathbb{Z}_{\geq 0}$. In other words, $Z_{p, i}$ contains elements that are spanned by $\spn(S\cup T_{\alpha}(B_{h - 1}))$ for at least $i$ choices $S$ out of $S_1^h, S_2^h, \dotsc, S_p^h$. For consistency, let us define $Z_{p, 0} \coloneq N$. We also need the shorthand $\gamma_{p,i} \coloneqq \E[r(Z_{p,i}\ |\ T_\alpha(B_{h - 1}))]$. Observe that since the definition of $B_h$ implies the inclusion $B_{h}\setminus T_\alpha(B_{h-1}) \subseteq Z_{q,\lceil (1-\varepsilon)\tau q \rceil}$, to prove \Cref{prop:progress}, it suffices to prove
\[
  \gamma_{q,\lceil (1-\varepsilon)\tau q \rceil} \leq e^{\frac{1}{2} \varepsilon^2 \tau q} \cdot r(N).
\]
We start with the following simple observation.

\begin{observation}\label{obs:relationBetweenZs}
For every $p \in [q]$ and integer $i$,
\begin{itemize}
	\item if $i > p$, then $Z_{p,i}=\emptyset$, and therefore, $\gamma_{p,i}=0$.
	\item if $p < q$ and $i \in [p]$, then
  \begin{equation*}
    Z_{p+1,i} = Z_{p,i} \cup (Z_{p,i-1} \cap \spn(S^h_{p+1}\cup T_\alpha(B_{h - 1}))).
  \end{equation*}
\end{itemize}
Additionally, the value of $\gamma_{p, 0}$ is independent of $p$, i.e., $\gamma_{p_1, 0} = \gamma_{p_2, 0}$ for every two values $p_1, p_2 \in [q]$.
\end{observation}
\begin{proof}
The first part of the observation follows immediately from the definition of $Z_{p,i}$, and the last part of the observation holds since
\[
	\gamma_{p_1, 0} = \E[r(Z_{p_1, 0}) \mid T_\alpha(B_{h - 1}))] = \E[r(N\setminus T_{\alpha}(B_{h-1}) \mid T_\alpha(B_{h - 1}))] = \E[r(Z_{p_2, 0}) \mid T_\alpha(B_{h - 1}))] = \gamma_{p_2, 0}
	.
\]

It remains to prove the second part of the observation. The set $Z_{p,i} \cup (Z_{p,i-1} \cap \spn(S^h_{p+1}\cup T_\alpha(B_{h - 1})))$ contains an element $e \in N\setminus T_{\alpha}(B_{h-1})$ if it obeys at least one of two conditions.
\begin{itemize}
	\item The first condition is that $e \in Z_{p,i}$, i.e., it is contained in $\spn(S \cup T_\alpha(B_{h - 1}))$ for at least $i$ many choices of set $S$ among $S^h_1, S^h_2, \dotsc, S^h_{p}$.
	\item The second condition is that $e \in Z_{p,i-1} \cap \spn(S^h_{p+1}\cup T_\alpha(B_{h - 1}))$, i.e., $e$ is contained in $\spn(S^h_{p+1}\cup T_\alpha(B_{h - 1}))$ and in $\spn(S \cup T_\alpha(B_{h - 1}))$ for at least $i - 1$ many choices of set $S$ among $S^h_1, S^h_2, \dotsc, S^h_{p}$.
\end{itemize}
One can verify that the elements that obey at least one of these conditions are exactly the elements that are contained in $\spn(S \cup T_\alpha(B_{h - 1}))$ for at least $i$ many choices of set $S$ among $S^h_1, S^h_2, \dotsc, S^h_{p + 1}$, and thus, they are exactly the elements of $Z_{p + 1,i}$.
\end{proof}

To bound the values $\gamma_{p,i}$, we use the following lemma, which is a consequence of the second bullet point of \Cref{lem:rankReductionInOverlap}, and shows that, when going from $p$ to $p+1$ samples, the total increase in the gammas is small for any postfix $i, i+1, \dotsc, p+1$ of the indices.
\begin{lemma}\label{lem:gammaRelation}
For every $p \in [q-1]$ and $i \in [p + 1]$,
  $
    \sum_{j=i}^{p+1} (\gamma_{p+1,j} - \gamma_{p,j}) \leq \alpha \gamma_{p,i-1}
  $.
\end{lemma}
\begin{proof}
  By the definition of $\gamma_{p, i}$,
	  \begin{align} \label{eq:gamma_sum}
    \sum_{j=i}^{p+1} (&\gamma_{p+1,j} - \gamma_{p,j})
       = \sum_{j=i}^{p+1} \E[r(Z_{p+1,j}\ |\ T_\alpha(B_{h - 1})) - r(Z_{p,j}\ |\ T_\alpha(B_{h - 1}))]\\\nonumber
       &= \sum_{j=i}^{p+1}\E[r(Z_{p+1,j} \setminus Z_{p,j}\ |\ Z_{p,j} \cup T_\alpha(B_{h - 1}))]\\\nonumber
       &\leq \sum_{j=i}^{p+1} \E[r(Z_{p,j-1} \cap \spn(S^h_{p+1}\cup T_\alpha(B_{h - 1}))\ |\ Z_{p,j} \cup T_\alpha(B_{h - 1}))]\\\nonumber
       &\leq \sum_{j=i}^{p+1} \E[r(Z_{p,j-1} \cap \spn(S^h_{p+1}\cup T_\alpha(B_{h - 1}))\ |\ (Z_{p,j} \cap \spn(S^h_{p+1}\cup T_\alpha(B_{h - 1}))) \cup T_\alpha(B_{h - 1}))],
  \end{align}
	where the second equality uses the inclusion $Z_{p,j} \subseteq Z_{p+1,j}$,
  the first inequality hold because \Cref{obs:relationBetweenZs} implies that $Z_{p+1,j} \setminus Z_{p,j} \subseteq Z_{p,j-1} \cap \spn(S^h_{p+1}\cup T_\alpha(B_{h - 1}))$,
  and the second inequality is due to submodularity of the rank function.
	
	Notice that since $Z_{p,j} \subseteq Z_{p,j-1}$, the summand on the rightmost side of \Cref{eq:gamma_sum} is equal to
	\begin{multline*}
		\E[r((Z_{p,j-1} \cap \spn(S^h_{p+1} \cup T_\alpha(B_{h - 1}))) \cup T_\alpha(B_{h - 1}))] \\- \E[r((Z_{p,j} \cap \spn(S^h_{p+1}\cup T_\alpha(B_{h - 1}))) \cup T_\alpha(B_{h - 1}))],
	\end{multline*}
	and thus, the sum on the rightmost side of \Cref{eq:gamma_sum} is telescopic. Collapsing this sum yields
  \begin{align*}
    \sum_{j=i}^{p+1} (\gamma_{p+1,j} - \gamma_{p,j})
       &\leq \E[(r(Z_{p,i-1} \cap \spn(S^h_{p+1}\cup T_\alpha(B_{h - 1}))) \cup T_\alpha(B_{h - 1}))] \\&\mspace{130mu}- \E[r((Z_{p,p+1} \cap \spn(S^h_{p+1}\cup T_\alpha(B_{h - 1}))) \cup T_\alpha(B_{h - 1}))]\\
       &= \E[r(Z_{p,i-1} \cap \spn(S^h_{p+1}\cup T_\alpha(B_{h - 1}))\ |\ T_\alpha(B_{h - 1}))] \\
       & \leq \alpha \cdot \E[r(Z_{p,i-1} \ |\ T_\alpha(B_{h - 1}))]
			 = \alpha \gamma_{p,i-1},
  \end{align*}
  where the second equality holds because \Cref{obs:relationBetweenZs} guarantees that $Z_{p,p+1}=\emptyset$,
  and the inequality is a consequence of the second bullet point of \Cref{lem:rankReductionInOverlap} applied with $Q=Z_{p,i-1}$.
\end{proof}

Recall that our goal is to upper bound $\gamma_{q,\lceil (1-\varepsilon)\tau q \rceil}$. Technically, we derive below such a bound using multiple applications of \Cref{lem:gammaRelation}. However, to make this derivation more structured and easier to understand, we define a potential function. This potential function is inspired by a proof for the well-known Chernoff bound. More precisely, given a parameter $\kappa \in \mathbb{R}_{>0}$ (to be set later), we define the potential function
\begin{equation*}
  \Phi_p(\kappa) \coloneqq \sum_{i=0}^{p} (\gamma_{p,i}-\gamma_{p,i+1})e^{\kappa\cdot i} \qquad \forall p\in \{0,\dots, q\}.
\end{equation*}

\begin{observation} \label{obs:boundThroughPotential}
It holds that $\gamma_{q,\lceil (1-\varepsilon)\tau q\rceil} \leq e^{-\kappa \cdot \lceil (1-\varepsilon)\tau q\rceil} \cdot \Phi_q(\kappa)$.
\end{observation}
\begin{proof}
Note that
\begin{align*}
\Phi_q(\kappa) &= \sum_{i=0}^{q} (\gamma_{q,i}-\gamma_{q,i+1})e^{\kappa\cdot i}\\
                &\geq \sum_{i=\lceil (1-\varepsilon)\tau q\rceil}^{q} (\gamma_{q,i}-\gamma_{q,i+1})e^{\kappa\cdot i}\\
                &\geq e^{\kappa \cdot \lceil (1-\varepsilon)\tau q\rceil} \cdot \sum_{i=\lceil (1-\varepsilon)\tau q\rceil}^{q} (\gamma_{q,i}-\gamma_{q,i+1})\\
                &= e^{\kappa \cdot \lceil (1-\varepsilon)\tau q\rceil} \cdot (\gamma_{q,\lceil (1-\varepsilon)\tau q\rceil} - \gamma_{q,q+1})\\
                &= e^{\kappa \cdot \lceil (1-\varepsilon)\tau q\rceil} \cdot \gamma_{q,\lceil (1-\varepsilon)\tau q\rceil},
\end{align*}
where the last equality holds by \Cref{obs:relationBetweenZs}.
\end{proof}

In light of Observation~\ref{obs:boundThroughPotential}, to upper bound $\gamma_{q,\lceil (1-\varepsilon)\tau q\rceil}$, it suffices to upper bound the potential $\Phi_q(\kappa)$.
We do this using the following lemma, which upper bounds the increase of $\Phi_p(\kappa)$ as a function of $p$.
The proof of \Cref{lem:potential_decrease} crucially leverages \Cref{lem:gammaRelation}.
\begin{lemma} \label{lem:potential_decrease}
For every $p \in [q]$ and $\kappa \in \mathbb{R}_{>0}$, we have $\Phi_p(\kappa) \leq (1 - \alpha + \alpha e^\kappa) \Phi_{p-1}(\kappa)$.
\end{lemma}
\begin{proof}
Let us define, for every $i\in \mathbb{Z}_{\geq -1}$,
\begin{equation*}
  \nu_i \coloneqq \begin{cases}
    0 & \text{if } i=-1,\\
    1 & \text{if } i=0,\\
    e^{\kappa \cdot i} - e^{\kappa\cdot(i-1)} & \text{if } i\geq 1.
  \end{cases}
\end{equation*}
Using this notation, we can rewrite $\Phi_p(\kappa)$ as
\begin{multline}\label{eq:potentialAsSlices}
  \Phi_p(\kappa) = \sum_{i=0}^{p} (\gamma_{p,i}-\gamma_{p,i+1})e^{\kappa\cdot i}
                  = \sum_{i=0}^p \gamma_{p,i} \cdot \nu_i\\
                  = \sum_{i=0}^{p} \bigg( \gamma_{p,i} \cdot \sum_{j=0}^i (\nu_j - \nu_{j-1})\bigg)
                  = \sum_{j=0}^p \bigg((\nu_j - \nu_{j-1}) \cdot \sum_{i=j}^p \gamma_{p,i}\bigg). 
\end{multline} 
Notice that we have used in the second equality the fact that $\gamma_{p, p+1} = 0$ by \Cref{obs:relationBetweenZs}.

We can now relate the potential $\Phi_p(\kappa)$ to $\Phi_{p-1}(\kappa)$ as follows (we elaborate on the steps below).
\begin{align*}
  \Phi_p(\kappa) &= \sum_{j=0}^p \bigg((\nu_j - \nu_{j-1}) \cdot \sum_{i=j}^p (\gamma_{p,i} - \gamma_{p-1,i} + \gamma_{p-1,i})\bigg)\\
   &= \sum_{j=0}^{p} \bigg((\nu_j - \nu_{j-1}) \cdot \sum_{i=j}^p \gamma_{p-1,i}\bigg) + \sum_{j=0}^{p} \bigg( (\nu_j - \nu_{j-1}) \cdot \sum_{i=j}^p (\gamma_{p,i} - \gamma_{p-1,i})\bigg)\\
   &= \sum_{j=0}^{p-1} \bigg((\nu_j - \nu_{j-1}) \cdot \sum_{i=j}^{p-1} \gamma_{p-1,i}\bigg)
       + \sum_{j=1}^{p} \bigg((\nu_j - \nu_{j-1}) \cdot \sum_{i=j}^p (\gamma_{p,i} - \gamma_{p-1,i})\bigg)
       + \sum_{i=1}^p(\gamma_{p,i} - \gamma_{p-1,i})\\
   &\leq \Phi_{p-1}(\kappa)
       + \sum_{j=1}^{p} (\nu_j - \nu_{j-1}) \cdot \alpha\gamma_{p-1,j-1}
       + \alpha \gamma_{p-1,0}\\
   &= \Phi_{p-1}(\kappa) - \alpha \cdot \sum_{j=0}^{p-1} \nu_j \gamma_{p-1,j}
   + \alpha \bigg(\gamma_{p-1,0} + \sum_{j=0}^{p-1} \nu_{j+1} \gamma_{p-1,j}\bigg)\\
   &= \Phi_{p-1}(\kappa) - \alpha \cdot\sum_{j=0}^{p-1} \nu_j \gamma_{p-1,j}
   + \alpha e^\kappa \sum_{j=0}^{p-1} \nu_{j} \gamma_{p-1,j}\\
   &= \left(1 - \alpha + \alpha e^\kappa \right) \Phi_{p-1}(\kappa).
\end{align*}
The first equality above follows from \Cref{eq:potentialAsSlices}.
The third equality uses two properties given by \Cref{obs:relationBetweenZs}. First, the equality $\gamma_{p-1,p}=0$ is used to decrease the upper limit of the first two summations from $p$ to $p-1$. Second, the third term after the third equality comes from splitting off the $j=0$ term of the second summation, and then using $\nu_0=1$, $\nu_{-1}=0$, and $\gamma_{p,0}=\gamma_{p-1,0}$ to increase the lower limit of the summation in the third term from $0$ to $1$.
The inequality uses \Cref{lem:gammaRelation} and again \Cref{eq:potentialAsSlices}.
The penultimate equality is based on the equalities $\nu_{j+1}= e^{\kappa} \cdot \nu_j$ for $j\in \mathbb{Z}_{\geq 1}$, and $\nu_1 + 1 = e^\kappa \cdot\nu_0$, which follow from the definition of the $\nu_j$'s.
The final equality follows from the second equality in \Cref{eq:potentialAsSlices}.
\end{proof}

We are now ready to prove \Cref{lem:BNotGrowingMuch}

\begin{proof}[Proof of \Cref{lem:BNotGrowingMuch}]
By applying \Cref{lem:potential_decrease} for every $p \in [q]$, we get
\[
  \Phi_q(\kappa)  \leq \big(1 - \alpha + \alpha e^\kappa \big)^q \cdot \Phi_0(\kappa)
  \leq \big(1 - \alpha + \alpha e^\kappa \big)^q \cdot r(N),
\]
where we used $\Phi_0(\kappa) = \gamma_{0,0} = \E[r(N \setminus T_{\alpha}(B_{h-1}) \ |\ T_{\alpha}(B_{h-1}))] = \E[r(N \ |\ T_{\alpha}(B_{h-1}))] \leq r(N)$ in the last inequality. Combining this inequality with \Cref{obs:boundThroughPotential}, we get the following upper bound on $\E[r(B_h\ |\ T_\alpha(B_{h-1}))]$, which holds for any $\kappa > 0$.
\begin{align*}
  \E[r(B_h\ |\ T_\alpha(B_{h-1}))] &\leq \gamma_{q,\lceil (1-\varepsilon)\tau q\rceil} \leq e^{-\kappa \cdot \lceil (1-\varepsilon)\tau q\rceil} \cdot \Phi_q(\kappa)\\
										 &\leq e^{-\kappa \cdot \lceil (1-\varepsilon)\tau q\rceil} \cdot \big(1 - \alpha + \alpha e^\kappa \big)^q  \cdot r(N)\\
                     &\leq e^{-\kappa \cdot  (1-\varepsilon)\tau q} \cdot e^{(e^\kappa -1) \cdot \alpha q} \cdot r(N)\\
                     &= \big(e^{(1 - 2\varepsilon)(e^\kappa -1) - (1-\varepsilon)\kappa}\big)^{\tau q} \cdot r(N).
\end{align*}
The last equality uses the fact that $\alpha = (1 - 2\varepsilon)\tau$ by definition~\eqref{eq:definitionalpha}.

We are now free to choose an arbitrary positive value for $\kappa$, and it turns out that the rightmost side of the above inequality is minimized for $\kappa = \ln(\frac{1-\varepsilon}{1-2\varepsilon})$. Plugging this value into the inequality yields the desired bound:
\begin{align*}
  \E[r(B_h\ |\ T_{\alpha}(B_{h-1}))] &\leq \left(\left(\frac{1-2\varepsilon}{1-\varepsilon}\right)^{(1-\varepsilon)} e^{\varepsilon}\right)^{\tau q} \cdot r(N)\\
                     &= \left(\left(1-\frac{\varepsilon}{1-\varepsilon}\right)^{(1-\varepsilon)} e^{\varepsilon}\right)^{\tau q} \cdot r(N)\\
                     &\leq e^{-\frac{1}{2} \frac{\varepsilon^2}{1-\varepsilon} \tau q} \cdot r(N)\\
                     &\leq e^{-\frac{1}{2} \varepsilon^2 \tau q} \cdot r(N),
\end{align*}
where the second inequality uses the inequality $1-z \leq e^{-z - \frac{1}{2}z^2} \;\;\forall z\in \mathbb{R}_{\geq 0}$, which holds due to the following.
The function $f(z) \coloneqq e^{-z - \frac{1}{2}z^2} + z -1$ is non-decreasing on $\mathbb{R}_{\geq 0}$ and fulfills $f(0)=0$.
Hence, for any $z\in \mathbb{R}_{\geq 0}$, we have $f(z) \geq f(0) = 0$, as desired.
\end{proof}

\section{Bounding Bad Events in Link} \label{sec:inlink_loss}

In this section, we prove \Cref{prop:inLinkLose}, which we repeat here for convenience.

\propInLinkLose*

As was explained at the start of Section~\ref{sec:progress}, we may pretend that the upper bound of the loop on line~\ref{algline:loopInLink} of \algsinglelink is $\eta$ rather than $\overline{h}$.
This does not affect the output of \algsinglelink, but keeps the randomness of $\overline{h}$ separate from the randomness of the samples takes in line~\ref{algline:getSamples} of the procedure, and allows us to assume that all the sets $A_0, A_1, \dots, A_{\eta}$ are always defined.

Fix an arbitrary element $e \in N$. We say that a set $A_h$ is \emph{critical} if $e$ is $(A_h,\tau)$-bad; i.e., $A_h$ is critical if $e\not\in A_h$ and
  \begin{equation*}
    \Pr[e\in \spn(A_h \cup R(x))] > \tau.
  \end{equation*}
Recall now that our aim is to prove \Cref{prop:inLinkLose}, which upper bounds the probability that $e$ is $(A, \tau)$-bad. Since $A$ is sampled out of $A_0, A_1, \dotsc, A_{\eta}$, this means that we need to upper bound the probability that the sampled set is critical. To simplify this task, we partition the probability space into the following disjoint events.
  \begin{itemize}
    \item $\mathcal{E}_{\mathrm{mult}}$ is the event that the sequence $A_0, \dots, A_{\eta}$ contains at least two critical sets.
    \item For $h\in \{0,\dots, \eta\}$, $\mathcal{E}_h$ is the event that $A_h$ is the only critical set in the sequence.
    \item $\mathcal{E}_\emptyset$ is the event that there is no critical set in the sequence.
  \end{itemize}
	
The following observation demonstrates why the events $\mathcal{E}_h$ are useful.
\begin{observation}
For every integer $0 \leq h \leq \eta$,
\begin{itemize}
	\item $\Pr[\text{$e$ is $(A, \tau)$-bad} \mid \mathcal{E}_h] = \Pr[\overline{h} = h \mid \mathcal{E}_h] = \Pr[\overline{h} = h]$.
	\item $\Pr[\text{$e$ is $(A, \tau)$-good} \mid \mathcal{E}_h] = \Pr[\overline{h} < h \mid \mathcal{E}_h] = \Pr[\overline{h} < h]$.
\end{itemize}
\end{observation}
\begin{proof}
The second equality in each of the above bullets holds since $\overline{h}$ is chosen independently of the values of the sets $A_0, A_1, \dotsc, A_{\eta}$, and therefore, also independently of $\mathcal{E}_h$. The equality $\Pr[\text{$e$ is $(A, \tau)$-bad} \mid \mathcal{E}_h] = \Pr[\overline{h} = h \mid \mathcal{E}_h]$ holds because the event $\mathcal{E}_h$ implies that $A_h$ is the only critical set among $A_0, A_1, \dotsc, A_{\eta}$. Thus, to prove the observation, it remains to show that the equality $\Pr[\text{$e$ is $(A, \tau)$-good} \mid \mathcal{E}_h] = \Pr[\overline{h} < h \mid \mathcal{E}_h]$ holds as well.

Since $A_h$ is critical under $\mathcal{E}_h$, we know that $e \not \in A_h$ and $\Pr[e \in A_h \cup R(x)] > \tau$. Notice now that it always holds that $A_0 \subseteq A_1 \subseteq \dotso \subseteq A_{\eta}$, and consider two cases.
\begin{itemize}
	\item If $\overline{h} < h$, then $e \not \in A_h \supseteq A_{\overline{h}}$. Since $A_{\overline{h}}$ is not a critical set under $\mathcal{E}_h$, this implies that $e$ must be $(A, \tau)$-good.
	\item If $\overline{h} \geq h$, then $\Pr[e \in A_{\overline{h}} \cup R(x)] \geq \Pr[e \in A_h \cup R(x)] > \tau$, and therefore, $e$ is not $(A, \tau)$-good.
\end{itemize}
Thus, we get the desired equality $\Pr[\text{$e$ is $(A, \tau)$-good} \mid \mathcal{E}_h] = \Pr[\overline{h} < h \mid \mathcal{E}_h]$.
\end{proof}
	
Using the last observation, we get
\begin{equation}\label{eq:critical}
\begin{aligned} 
	\Pr[\text{$e$ is $(A, \tau)$-bad}]
	   &\leq
	\Pr[\mathcal{E}_\emptyset] \cdot 0 + \sum_{h = 0}^{\eta} \Pr[\mathcal{E}_h] \cdot \Pr[\overline{h} = h \mid \mathcal{E}_h] + \Pr[\mathcal{E}_\mathrm{mult}] \cdot 1 \\
	   &= \sum_{h = 0}^{\eta} \Pr[\mathcal{E}_h] \cdot \Pr[\overline{h} = h] + \Pr[\mathcal{E}_\mathrm{mult}]\\
     &\leq \sum_{h=0}^{\eta} \Pr[\mathcal{E}_h] \cdot
        \bigg(\varepsilon \cdot \Pr[\overline{h} < h]  + \frac{\varepsilon^3}{\ln \rho}\bigg) + \Pr[\mathcal{E}_\mathrm{mult}]\\
	   &\leq
	\varepsilon \cdot \sum_{h=0}^{\eta} \Pr[\mathcal{E}_h] \cdot \Pr[\overline{h} < h \mid \mathcal{E}_h] + \frac{\varepsilon^3}{\ln \rho} + \Pr[\mathcal{E}_\mathrm{mult}]\\
     &\leq
	\varepsilon \cdot \Pr[\text{$e$ is $(A,\tau)$-good}] + \frac{\varepsilon^3}{\ln \rho} + \Pr[\mathcal{E}_\mathrm{mult}]
	,
\end{aligned}
\end{equation}
where the first and last inequalities hold by the law of total probability, and the second inequality holds by the next observation.

\begin{observation} \label{obs:bounding_single}
For every integer $0 \leq h \leq \eta$,
  \[
    \Pr[\overline{h} = h]
    \leq \varepsilon \cdot \Pr[\overline{h} < h] + \frac{\varepsilon^3}{\ln \rho}.
  \]
\end{observation}
\begin{proof}
The observation clearly holds for $h=0$ because $\overline{h}$ is always at least $1$, and thus, $\Pr[\overline{h}=0]=0$. Additionally, for $h = 1$, since $\eta = \big\lceil 1 + \log_{1+\varepsilon} \frac{\ln \rho}{\varepsilon^3} \big\rceil$,
\[
	\Pr[\overline{h} = 1]
	=
	\frac{1}{(1+\varepsilon)^{\eta - 1}}
	\leq
	\frac{\varepsilon^3}{\ln \rho}.
\]
Finally, for $h \geq 2$, by the definition of the distribution of $\bar{h}$,
\[\Pr[\overline{h} = h] = \varepsilon \cdot \Pr[\overline{h} < h] < \varepsilon \cdot \Pr[\overline{h} < h] + \frac{\varepsilon^3}{\ln \rho}.\tag*{\qedhere}\]
\end{proof}


\Cref{prop:inLinkLose} now follows by plugging into \Cref{eq:critical} the bound on $\Pr[\mathcal{E}_\mathrm{mult}]$ proved by the next lemma.

\begin{lemma} \label{lem:bounding_multiple}
It holds that $\Pr[\mathcal{E}_\mathrm{mult}] \leq \frac{\varepsilon^3}{\ln \rho}$.
\end{lemma}

\begin{proof}
Recall that $A_0 \subseteq A_1 \subseteq \dotso \subseteq A_{\eta}$. This implies that the critical sets form a continuous sub-sequence of the sequence $A_0, A_1, \dotsc, A_{\eta}$.
  Hence, $\mathcal{E}_{\mathrm{mult}}$ happens if and only if there exists an index $h$ such that $A_h$ is the first critical set in the sequence and $A_{h + 1}$ is also critical. 
  However, we (implicitly) show below that if some set $A_h$ is critical, then the next set $A_{h+1}$ is very likely to contain $e$, and thus, very unlikely to be critical.

Let $\mathcal{F}_{h}$ be the event that $A_h$ is the first critical set. By the above discussion,
\[
	\Pr[\mathcal{E}_{\textrm{mult}}]
	=
	\sum_{h = 0}^{\eta - 1} \Pr[\mathcal{F}_h] \cdot \Pr[\text{$A_{h + 1}$ is critical} \mid \mathcal{F}_h]
	.
\]
Thus, to prove the lemma, is suffices to show that $\Pr[\text{$A_{h + 1}$ is critical} \mid \mathcal{F}_h] \leq \frac{\varepsilon^3}{\ln n}$ for every integer $0 \leq h \leq \eta - 1$. In the rest of the proof, we implicitly condition on $\mathcal{F}_h$. Let $R_1, R_2, \dots, R_q$ be the $q$ independent samples from the distribution $\mathcal{D}(x)$ used in the construction of $A_{h+1}$. It is important to notice that the distribution of the samples $R_1, R_2, \dots, R_q$ is unaffected by the conditioning on $\mathcal{F}_h$ because it is possible to determine whether $\mathcal{F}_h$ happened even without knowing the values of these samples.

  For every $j\in [q]$, we denote by $Y_j$ a Bernoulli random variable that is $1$ if $e\in \spn(A_h \cup R_j)$ and $0$ otherwise.
  By the definition of $A_{h+1}$,
  \begin{equation*}
    e \not\in A_{h+1} \iff \frac{1}{q}\sum_{j=1}^q Y_j \leq (1-\varepsilon)\tau.
  \end{equation*}
  In contrast, $A_h$ being critical yields
  \begin{equation*}
    \Pr[Y_j = 1] > \tau \quad \forall j\in [q],
  \end{equation*}
  which, by using a classical Chernoff bound, implies
  \begin{equation*}
    \Pr[A_{h+1} \text{ is critical}] \leq \Pr[e\not\in A_{h+1}] = \Pr\left[\frac{1}{q}\sum_{j=1}^q Y_j \leq (1-\varepsilon) \tau \right]
         \leq e^{-\frac{\varepsilon^2 q \tau}{2}}
         \leq \left(\frac{\varepsilon}{\ln \rho}\right)^3 \leq \frac{\varepsilon^3}{\ln \rho}.
  \end{equation*}
(The first inequality holds, in fact, as an equality since $A_h$ being critical implies that $A_{h + 1}$ is critical whenever $e \not \in A_{h + 1}$.)
\end{proof}

\printbibliography

\end{document}